\documentclass{CSML}
\pdfoutput=1

\usepackage{lastpage}
\newcommand{\dOi}{14(1:18)2018}
\lmcsheading{}{1--\pageref{LastPage}}{}{}%
{Apr.~29, 2016}{Feb.~27, 2018}{}

\usepackage{hyperref}
\hypersetup{hidelinks}

\usepackage{MnSymbol}
\usepackage{stmaryrd}

\usepackage{prooftree}
\usepackage{csquotes}

\usepackage{xcolor}
\usepackage{adjustbox}

\newcommand{\record}[1]{\ensuremath{{\langle #1 \rangle}}}
\newcommand{\rrecord}[1]{\ensuremath{{\llangle #1 \rrangle}}}
\newcommand{\interpret}[1]{\ensuremath{{\llbracket #1 \rrbracket}}}

\newcommand{\pipe}[0]{\ensuremath{\gg}}
\newcommand{\vdashr}[0]{\ensuremath{\vdash_{\record{}}}}

\newcommand{\vdashk}[0]{\ensuremath{\vdash_k}}

\newcommand{\DeltaLCM}[0]{\ensuremath{\Delta_{\mathcal{L}}^{\mathcal{C}, \mathcal{M}}}}

\newcommand{\Der}{\mathcal{D}}

\newcommand{\myClass}[0]{\textsf{myClass}}
\newcommand{\argClass}[0]{\textsf{argClass}}
\newcommand{\state}[0]{\textsf{state}}

\newcommand{\self}[0]{\textsf{self}}
\newcommand{\super}[0]{\textsf{super}}

\newcommand{\get}[0]{\textit{get}}
\newcommand{\set}[0]{\textit{set}}

\renewcommand{\succ}[0]{\textit{succ}}
\newcommand{\succTwice}[0]{\textit{succ2}}
\newcommand{\compare}[0]{\textit{compare}}

\newcommand{\semType}[1]{{\color{blue}\ensuremath{\textsf{#1}}}}
\newcommand{\Int}[0]{\text{Int}}
\newcommand{\String}[0]{\text{String}}
\newcommand{\Bool}[0]{\text{Bool}}

\newcommand{\Even}[0]{\semType{Even}}
\newcommand{\Odd}[0]{\semType{Odd}}

\newcommand{\Nat}[0]{\text{Num}}
\newcommand{\Comparable}[0]{\text{Comparable}}
\newcommand{\Parity}[0]{\text{Parity}}

\newcommand{\SuccTwice}[0]{\text{Succ2}}
\newcommand{\SuccDelta}[0]{\text{SuccDelta}}

\newcommand{\Injective}[0]{\semType{Injective}}

\newcommand{\Plain}[0]{\semType{Plain}}
\newcommand{\Enc}[0]{\semType{Enc}}
\newcommand{\Sign}[0]{\semType{Sign}}
\newcommand{\Time}[0]{\semType{Time}}

\DeclareMathOperator{\tgt}{tgt}

\DeclareMathOperator{\atoms}{atoms}
\DeclareMathOperator{\level}{level}

\newcommand{\Y}[0]{\mathbf{Y}}
\newcommand{\letin}[1]{\texttt{let } #1 \texttt{ in }}

\newcommand{\TT}[0]{\mathbb{T}}

\newcommand{\CC}[0]{\mathbb{C}}
\newcommand{\NN}[0]{\mathbb{N}}
\renewcommand{\AA}[0]{\mathbb{A}}
\newcommand{\PP}[0]{\mathbb{P}}
\newcommand{\calL}[0]{\mathcal{L}}
\newcommand{\calO}[0]{\mathcal{O}}
\newcommand{\calC}[0]{\mathcal{C}}
\newcommand{\calM}[0]{\mathcal{M}}

\newcommand{\TTC}[0]{\TT_{\CC}}
\newcommand{\TTR}[0]{\TT_{\record{}}}
\newcommand{\bclk}{\ensuremath{\textsf{BCL}_k(\to,\cap)}}
\newcommand{\bcl}{\ensuremath{\textsf{BCL}(\to,\cap)}}
\newcommand{\bclc}{\ensuremath{\textsf{BCL}(\TTC)}}
\newcommand{\bclkc}{\ensuremath{\textsf{BCL}_k(\TTC)}}

\newskip \point \point=1pt

\newcommand{\Then}{\Rightarrow} 
\newcommand{\Iff}{\Leftrightarrow} 

\newcommand{\Set}[1]{\{ #1 \}} 

\newcommand{\dom}{\textrm{\it dom}}

\newcommand{\Pair}[1]{\langle #1 \rangle} 

\newcommand{\reduces}{\longrightarrow} 

\newcommand{\Label}{\mbox{\bf Label}}
\newcommand{\Variable}{\mbox{\bf Var}}
\newcommand{\lbl}{\mbox{\it lbl}}

\newcommand{\lambdaR}{\Lambda_R}

\newcommand{\inter}{\cap} 

\newcommand{\der}{\vdash} 
\newcommand{\deduces}{\vdash} 

\newcommand{\arrow}{\rightarrow} 

\newcommand{\ArrI}{\arrow\mbox{\it I}}   
\newcommand{\ArrE}{\arrow\mbox{\it E}} 
\newcommand{\Axiom}{\textrm{\it Ax}} 

\newcommand{\redSel}{\mbox{\it sel}}

\newcommand{\redMergeThree}{\mbox{\it $\oplus$}}
\newcommand{\recRule}{\mbox{\it rec}}
\newcommand{\selRule}{\mbox{\it sel}}

\newcommand{\Subst}[2]{\{#1 / #2\}}

\newcommand{\fv}{\mbox{\it fv}}

\newcommand{\Override}{\,\oplus\,}

\begin{document}

\title[Mixin Composition Synthesis]{Mixin Composition Synthesis based on Intersection Types}


\author[Bessai]{Jan Bessai\rsuper{a}}	
\address{\lsuper{a}Technical University of Dortmund, 
  Dortmund, Germany}	
\email{jan.bessai@tu-dortmund.de}  
\email{andrej.dudenhefner@tu-dortmund.de} 
\email{boris.duedder@tu-dortmund.de}  
\email{jakob.rehof@tu-dortmund.de}  

\author[Chen]{Tzu-Chun Chen\rsuper{b}}	
\address{\lsuper{b}Technical University of Darmstadt, Darmstadt, Germany}	
\email{tcchen@rbg.informatik.tu-darmstadt.de}

\author[Dudenhefner]{Andrej Dudenhefner\rsuper{a}}	
\address{\vskip-6pt}

\author[D{\"u}dder]{Boris D{\"u}dder\rsuper{a}}	
\address{\vskip-6pt}

\author[de'Liguoro]{Ugo de'Liguoro\rsuper{c}}	
\address{\lsuper{c}University of Torino, Torino, Italy}	
\email{ugo.deliguoro@unito.it}

\author[Rehof]{Jakob Rehof\rsuper{a}}	
\address{\vskip-6pt}

\thanks{This work was partially supported by EU COST Action IC1201: BETTY and MIUR PRIN CINA Prot.\ 2010LHT4KM, San Paolo Project SALT.
Tzu-chun Chen is also partially supported by the ERC grant FP7-617805 \emph{LiVeSoft -- Lightweight Verification of Software}.}




\titlecomment{This article is based on the submission to the 13th International Conference on Typed Lambda Calculi and Applications, TLCA 2015}



\begin{abstract}
  \noindent We present a method for synthesizing compositions of mixins using type inhabitation in intersection types. First, recursively defined classes and mixins, which are functions over classes, are expressed as terms in a lambda calculus with records. Intersection types with records and record-merge are used to assign meaningful types to these terms without resorting to recursive types. 
	Second, typed terms are translated to a repository of typed combinators.
 We show a relation between record types with record-merge and intersection types with constructors. This relation is used to prove soundness and partial completeness of the translation with respect to mixin composition synthesis.  
  Furthermore, we demonstrate how a translated repository and goal type can be used as input to 
  an existing framework for composition synthesis in bounded combinatory logic via type inhabitation. The computed result is a class typed by the goal type and generated by a mixin composition applied to an existing class.
\end{abstract}

\maketitle

  \section{Introduction}
	

Starting with Cardelli's pioneering work \cite{Cardelli84}, various typed $\lambda$-calculi extended with records have been thoroughly studied to model sophisticated features of object-oriented programming languages, like recursive objects and classes, object extension, method overriding and inheritance (see e.g.~\cite{AbadiCardelli,BruceBook,KiselyovLaemmel05}).

Here, we focus on the synthesis of mixin application chains. In the object-oriented paradigm, mixins \cite{Moon86, Cannon79} have been introduced as an alternative construct for code reuse that improves over the limitations of multiple inheritance, e.g.~connecting incompatible base classes and semantic ambiguities caused by the diamond problem \cite{HTT87}. 
Together with abstract classes and traits, mixins (functions over classes) can be considered as an advanced construct to obtain flexible implementations of module libraries and to enhance code reusability; many popular programming languages miss native support for mixins, but they are an object of intensive study and research (e.g.~\cite{BonoPS99,OderskyZ05}). 
In this setting we aim at synthesizing classes from a library of mixins.
Our particular modeling approach is inspired by modern language features (e.g.~ECMAScript ``\texttt{bind}'') to preserve contexts in order to prevent programming errors \cite{ecma2011rev51}.

We formalize synthesis of classes from a library of mixins as an instance of the relativized type inhabitation problem in combinatory logic. Relativized type inhabitation is the following decision problem: given a combinatory type context $\Delta$ and a type $\tau$ does there exist an applicative term $e$ such that $e$ has type $\tau$ under the type assumptions in $\Delta$? We implicitly include the problem of constructing such a term as the synthesized result.

Although relativized type inhabitation is undecidable even in simple types, it is decidable in $k$-bounded combinatory logic $\bclk$~\cite{RehofEtAlCSL12}. $\bclk$ is the combinatory logic typed with arrow and intersection types, where $k$ bounds  the depth of types wrt. $\to$ used to instantiate schematic combinator types.
Hence, an algorithm for semi-deciding type inhabitation for $\mbox{\bcl} = \bigcup_k \mbox{\bclk}$ can be obtained by iterative deepening over $k$ and solving the corresponding decision problem in \bclk\; \cite{RehofEtAlCSL12}. In the present paper, we 
enable combinatory synthesis of classes via intersection typed mixin combinators. Intersection types \cite{bcd} play an important 
r\^{o}le in combinatory synthesis, because they allow for semantic specification of components and synthesis goals \cite{RehofEtAlCSL12,BessaiDDM14}. They also allow a natural way to type records.

Now, looking at $\Set{C_1:\sigma_1,\ldots,C_p:\sigma_p, M_1:\tau_1,\ldots,M_q:\tau_q } \subseteq \Delta$
as the abstract specification of a library including classes $C_i$ and mixins $M_j$ with interfaces $\sigma_i$ and $\tau_j$ respectively,
and given a type $\tau$ specifying a desired class, we
may identify the class synthesis problem with the relativized type inhabitation problem of constructing a term $e$, i.e.~an applicative composition of classes and mixins, typed by $\tau$ in $\Delta$.
To make this feasible, we have to bridge the gap between the expressivity of highly sophisticated type systems used for typing classes and mixins, for instance
$F$-bounded polymorphism used in \cite{CHOM89,CookHC90}, and the system of intersection types from \cite{bcd}.
In doing so, 
we move from the system originally presented in \cite{deLiguoro01}, consisting of a type assignment system of intersection and record
types, to a $\lambda$-calculus which we enrich here with record merge operation (called ``\textbf{with}'' in  \cite{CookHC90}), to allow for expressive mixin combinators.
The type system is modified by reconstructing record types $\record{l_i:\sigma_i \mid i \in I}$ as intersection of unary record types $\record{l_i:\sigma_i}$, and considering a subtype relation extending the one in \cite{bcd}. This is however not enough for typing record merge, for which we consider a type-merge operator $+$. The problem of typing extensible records and merge, faced for the first time in \cite{Wand91,Remy92}, is notoriously hard; to circumvent difficulties the theory of record subtyping in \cite{CookHC90} (where a similar type-merge operator is considered) allows just for ``exact'' record typing, which involves subtyping in depth, but not in width. Such a restriction, that has limited effects wrt. a rich and expressive type system like $F$-bounded polymorphism, would be too severe in our setting. Therefore, we undertake a study of the type algebra of record types with intersection and type-merge, leading to a type assignment system where exact record typing is required only for the right-hand side operand of the term merge operator, which is enough to ensure soundness of typing.

The next challenge is to show that in our system we can type classes and mixins in a meaningful way. Classes are essentially recursive records. Mixins are made of a combination of fixed point combinators and record merge. Such combinators, which usually require recursive types, can be typed
in our system by means of an iterative method exploiting the ability of intersection types to represent approximations of the potentially infinite unfolding of recursive definitions. 

The final problem we face is the encoding of intersection types with record types and type-merge into the language of $\bclk$. For this purpose, we consider a conservative extension of bounded combinatory logic, called $\bclkc$, where we allow unary type constructors that are monotonic and distribute over intersection. We show that the (semi) algorithm solving inhabitation for $\bclk$ can be adapted to $\bclkc$, by proving that the key properties necessary to solve the inhabitation problem in $\bclk$ are preserved in $\bclkc$ and showing how the type-merge operator can be simulated in $\bclkc$. In fact, type-merge is not monotonic in its second argument, due to the lack of negative information caused by the combination of $+$ and $\inter$. Our work culminates in two theorems that ensure soundness and completeness of the so obtained method wrt. synthesis
of classes by mixin application.

\subsection{Contributions}
The contributions of this article can be summarized as follows:
\begin{itemize}
\item A type system with intersection types and records ($\TTR$) for a $\lambda$-calculus with records ($\lambdaR$) is designed and its key properties are proven. 
\item $\TTR$ is used as a typed calculus for classes and mixins.
\item Bounded combinatory logic and a decision algorithm for relativized type inhabitation are extended with constructors ($\bclkc$) retaining complexity results.
\item A sound and (partially) complete encoding of mixins and classes in $\TTR$ as combinators in $\bclkc$ for synthesis is proven and exemplified.
\item Negative information, i.e.~the information on which labels are absent, is encoded with polynomial overhead.
\end{itemize}

\subsection{Organization}
The article is organized as follows: 
Section~\ref{sec:relatedWork} discusses the development of traits and mixins as well as program synthesis. In Section~\ref{sec:intersection}, intersection and record types for a $\lambda$-calculus with records are introduced as a domain specific language for representing mixins and classes. Section~\ref{sec:bcl} adapts bounded combinatory logic to include constructors ($\bclkc$) as a foundation for mixin synthesis. Section~\ref{sec:synthesis} presents the encoding of record types in $\bclkc$, mixin composition synthesis by type inhabitation and provides detailed examples including the use of semantic types. A conclusion in Section~\ref{sec:conclusion} hints at future work.

	\section{Related Work}\label{sec:relatedWork}
	

This work evolves  from the contributions \cite{BessaiDDM14, UdLTC2014, besDDTU15} to the workshop \textsf{ITRS'14} and the original submission to TLCA15 \cite{BDDCdLR15}. It combines two research areas, type systems for traits and mixins and program synthesis using type inhabitation.

\subsection{Type Systems for Traits and Mixins}
The concept of mixins goes back to research on Lisp dialects in the late 70's and early 80's \cite{Cannon79, WM80, Moon86}. The motivation was to extend code organization in object-oriented languages, which was traditionally defined by hierarchical inheritance. To this end Cannon introduced base flavors serving as a starting point later to be extended by mixin flavors \cite{Cannon79}. Mixin flavors implement particular features and can only be instantiated by application to existing flavors. They can define requirements to be met prior to application. If they use internal state, this state cannot be shared with other flavors. 

Type theories for objects and classes were developed shortly after. Languages like Simula~\cite{DN66} and Smalltalk-80~\cite{Deutsch84} are object-oriented typed languages, but their type system is geared toward memory management rather than safe reuse. 
Various techniques have been deployed to enable type safe reuse of implementations. Cardelli's seminal paper ``A Semantics of Multiple Inheritance''~\cite{Cardelli84} filled this gap. Type inference for products with subtyping was shown to be decidable in \cite{Mitchell84}. Cardelli improved upon the cumbersome encoding of entities using records with named projections instead of products. Subtyping addresses  permutation of record entries as well as multiple inheritance. It ensures compatibility of horizontally extended and vertically specialized records. 

A series of later developments dealt with issues introduced by recursion and obtained type inference. To this end Wand designed an ML-inspired system \cite{Wand87} that later inspired the development of System $F_{\leq}$~\cite{CHOM89}. A fine grained analysis of the connection between inheritance and subtyping was now possible \cite{CookHC90} and came to the surprising negative conclusion, that ``inheritance is not subtyping''. This analysis by Cook, Hill and Canning is also the first to discuss the difference between early and late binding of the recursion inherently present in objects: delegating recursion to methods using a self parameter (late binding) as done previously by Mitchell \cite{Mitchell89} unifies classes and objects, whereas early fixpoint computation at the class level is open for later modifications by inheritance and mixins. Subsequently, Bracha and Cook reinterpreted mixins in the light of classes and inheritance, where they can act as abstract subclasses \cite{BrachaC90}. Their main contribution was to shed light on the abstract operation of mixin application in contrast to prior work, which focused on implementation details like class hierarchy linearization. This operation turns out to be powerful enough to encode (multiple) inheritance.
From there on, development took two main routes. One direction was the development of advanced object calculi, which is covered in great detail in \cite{AbadiCardelli}. The other direction was to model the type systems of mainstream programming languages. Featherweight Java by Igarashi, Pierce and Wadler is a prominent example \cite{IgarashiPW01} for the latter direction. It inspired a mixin based treatment of virtual classes in \cite{EOC06}. A parallel development focused on traits, which are a restriction of mixins to subtype compatible overwriting \cite{LS08, LSZ12}. Bono et al. suggested a formal type system for a calculus with traits \cite{BDG07, BDG08}. The most recent studies follow this trend to bring together the two directions. Java has been partially formalized by Rowe and van Bakel \cite{BR13, RoweB14}, and completely in K-Java~\cite{BR15}. Recently, the Scala compiler is being reworked to match DOT~\cite{AMO12, ARO14, AGORS16}, a formal dependently typed core specification.

\subsection{Synthesis by Type Inhabitation}
Synthesis is by now a vast area of computer science, and we can only attempt here to place our approach broadly in relation to major points of comparison.

The synthesis problem can be traced back to Alonzo Church \cite{Church57, Thomas09}. It was first considered for the problem of automatically constructing a finite-state procedure implementing a given input/output relation over infinite bitstreams specified as a logical formula. The automata-theoretic approach was notably advanced by Pnueli and Rosner in the context of linear temporal logic (LTL) \cite{PR89} and by many others since then. Another major branch of work in synthesis is deductive synthesis,
as studied by Manna and Waldinger, who rely on proof systems for program properties \cite{MW80}, 
with many more works following in this tradition. 

The presented approach to the synthesis problem is proof-theoretic and
can best be characterized more specifically as {\em type-based} and 
{\em component-oriented}.
This approach was initiated by Rehof et al. \cite{RehofUrzyczyn11, RehofEtAlCSL12, JR13} around the idea of using the inhabitation problem in bounded combinatory logics with intersection types \cite{bcd} as a foundation for synthesis. In contrast to standard
combinatory logic \cite{HinSel08,hide92} invented by Sch{\"o}nfinkel in 1924 \cite{Schoenfinkel24}, $k$-bounded combinatory logic $\bclk$ \cite{RehofEtAlCSL12} introduces
a bound $k$ on the level of
types, i.e.~depth wrt. $\to$, used to instantiate schematic combinator types. Additionally, rather than
considering a {\em fixed} base of combinators (for example, the base
$\mathbf{S}, \mathbf{K}$), 
the inhabitation problem is {\em relativized} to an arbitrary
set $\Delta$ of typed combinators, given as part of the input to the relativized inhabitation problem:
\smallskip

\hfil {\em Given $\Delta$ and  $\tau$, is there 
an applicative term $e$ such that $\Delta \vdash_{\bclk} e : \tau$?}%
\smallskip 

\noindent In many instances the inhabitation problem with a fixed-base is much easier than in the relativized case, where the base is part of the input. 
For example, {\sc Pspace}-completeness of inhabitation in the simple-typed $\lambda$-calculus \cite{statman79} implies {\sc Pspace}-completeness of the equivalent simple-typed $\mathbf{S}\mathbf{K}$-calculus. 
The related problem of term enumeration for simple-typed $\lambda$-calculus has been studied by Ben-Yelles \cite{BY79} and, more recently, by Hindley~\cite{Hindley08}.

Linial and Post \cite{PostLinial} initiated the study of decision problems for arbitrary propositional axiom systems (partial propositional calculi, abbreviated PPC), in reaction to a question posed by Tarski in 1946. In the Linial-Post theorem they show existence of a PPC with an unsolvable decision problem. 
Later, Singletary showed that every recursively enumerable many-one degree can be represented by the implicational fragment of PPC \cite{Singletary74}. This implies that the relativized inhabitation problem is undecidable for combinatory logic with schematism even in simple types. Recent developments \cite{Bokov15} show that it is undecidable, whether a given finite set of propositional formulas constitutes an adequate axiom system for a fixed intuitionistic implicational propositional calculus.

In contrast, the main result of \cite{RehofEtAlCSL12} is that the relativized inhabitation problems for $\bclk$ with intersection types form an infinite hierarchy, being $(k+2)$-{\sc Exptime}-complete for each fixed bound $k$.

In \cite{JR13}, unbounded relativized inhabitation is seen, already at the level of simple types, 
to constitute a Turing-complete logic programming language for program composition,
in a way that is related to work on proof-theoretic generalizations of logic programming languages \cite{MillerNPS91}.

Perhaps surprisingly, component-oriented synthesis is a relatively recent development.
The approach of combinatory logic synthesis is basically motivated by the idea that type structure
provides a natural and code-oriented vehicle for synthesis specifications together with the
idea that combinatory logic provides a natural type-theoretic model of components, component interfaces, and component composition.
Vardi and Lustig initiated the component-oriented
approach within the automata-theoretic tradition, explicitly developing the idea of synthesizing systems from preexisting components \cite{LustigVardi09} with the idea of leveraging 
design intelligence and efficiency from components as building blocks for synthesis.
A recent Dagstuhl meeting explored these ideas of design and synthesis from components across
different communities \cite{RehofVardi14}.

Realizing the component-based idea within type-based synthesis requires ways to deal with 
semantic specification as well as the
combinatorial explosion of search spaces inherent in synthesis problems. 
Rehof et al.~\cite{RehofU12} introduced the idea of using intersection types \cite{bcd} to type-based synthesis
as a means of addressing the problem of search control and semantic specification at the type level. We refer to types enriched with semantic specifications through intersection types as {\em semantic types}.
In addition to semantic types, Rehof et al. have introduced the idea of {\em staging} into synthesis via modal types \cite{DudderMR14}.
Simple types were used by Steffen et al. \cite{SMB97} in the context of temporal logic synthesis
to semantically enrich temporal specifications with taxonomic information.

The notion of composition synthesis using semantic types is related to adaptation synthesis via proof counting \cite{Wells02,Wells04}. 
In particular, in \cite{Wells02} typed predicate logic is used for semantic specification at the interface level. We follow this idea, however the type system, underlying logic and algorithmic methods are different.

 Refinement types \cite{FP91} externally relate specifications to implementation types. The refinement type scheme structurally constraints how refinement types are formed: this prevents semantic specifications like $(\Int \to \Int) \cap \Injective$, which are possible in our system.
 Recently, refinement types have been used for example guided synthesis in \cite{FOW16}. 
Types have also been used in order to synthesize code completions \cite{Gkkp13}. Intersection types are not only useful for semantic specification, but also to encode objects. To this end they can be combined with records as proposed in \cite{Pierce91}. Their relation to object-oriented inheritance has been studied in \cite{CompagnoniP96} and serves as an inspiration for our work.

Combinatory logic synthesis has been implemented in a framework, Combinatory Logic Synthesizer {\bf (CL)S}, 
which is still being further developed \cite{BDDMR14}.

	\section{Intersection Types for Mixins and Classes}\label{sec:intersection}
	
\newcommand{\Context}{\Gamma}
\newcommand{\State}{\textsf{state}}

\subsection{Intersection and record types}

We consider a type-free $\lambda$-calculus of extensible records, equipped with a  merge operator.
The term syntax is defined by the following grammar:
\[
\begin{array}{rlll}
 \lambdaR \ni M,N, M_i  & ::= &  x \mid (\lambda x.M) \mid (MN)  \mid (M.l) \mid R \mid (M \oplus R) & \text{terms} \\ [1mm]
R & ::= & \record{ l_i = M_i \mid i \in I} & \text{records}
\end{array}
\]
where $x\in\Variable$ and $l\in\Label$ range over denumerably many {\em term variables} and {\em labels} respectively, and the sets of indexes $I$ are finite.
Free and bound variables are defined as for the ordinary $\lambda$-calculus, and we name $\lambdaR^0$ the set of all closed terms
in $\lambdaR$; terms are identified up to renaming of bound variables and $M\Subst{N}{x}$ denotes capture avoiding substitution of 
$N$ for $x$ in $M$. We adopt notational conventions from \cite{Barendregt84}; in particular, application associates to the left and external 
parentheses are omitted when unnecessary; also the dot notation for record selection takes precedence over $\lambda$, so that
$\lambda x. \;M.l$ reads as $\lambda x. (M.l)$. If not stated otherwise $\Override$ also associates to the left, and we avoid external parentheses when 
unnecessary.

Terms $R \equiv \record{ l_i = M_i \ | \ i \in I }$ (writing $\equiv$ for syntactic identity) represent {\em records}, 
with fields $l_i$ and $M_i$ as the respective values; we set $\lbl(\record{ l_i = M_i \mid i \in I}) = \Set{l_i \mid i\in I}$. For records we adopt the usual notation to combine sets, i.e.~$\{l_i, l_j \mid i \in I, j \in J\} = \{l_i \mid i \in I \cup J\}$. The term
$M.l$ is {\em field selection} and $M\oplus R$ 
is record {\em merge}. In particular, if $R_1$ and $R_2$ are records then $R_1\Override R_2$ is the record having as fields
the union of the fields of $R_1$ and $R_2$ and as values those of the original records but in case of ambiguity, where the values in
$R_2$ prevail. The syntactic constraint that $R$ is a record in $M \oplus R$ is justified after Definition \ref{def:typeAssignment}. Note that
a variable $x$ is not a record, hence $\lambda x. x \Override R$ is well-formed for any record $R$ while $\lambda x. R \Override x$ is not.

The actual meaning of these operations is formalized by the following reduction relation:

\begin{defi}[$\lambdaR$ reduction]
Reduction $\reduces\, \subseteq \lambdaR^2$ is the least compatible relation such that:
\[\begin{array}{lrcl}
(\beta) & \ (\lambda x.M)N &\reduces & M\Subst{N}{x} \\ [1mm]
(\redSel) & \Pair{l_i=M_i  \mid  i\in I}.l_j & \reduces & M_j  \hspace{4cm} \mbox{if $j \in I$} \\ [1mm]
(\redMergeThree) & \Pair{l_i=M_i  \mid  i\in I} \Override \Pair{l_j=N_j  \mid \ j\in J} & \reduces &
	\Pair{l_i=M_i, \ l_j=N_j \mid  i\in I\setminus J, \  j\in J}
\end{array}\]
\end{defi}

Record merge subsumes field update as considered in \cite{deLiguoro01}: 
$(M.l := N)$ is exactly $(M \Override \record{l = N})$, but merge is not uniformly definable in terms of update as long as
labels are not expressions in the calculus, therefore we take merge as primitive operator. The reduction relation
$\reduces$ is Church-Rosser namely its transitive closure $\reduces^*$ is confluent.

\begin{thm}[Church-Rosser property]
For all $M$ if $M \reduces^* M_1$ and $M \reduces^* M_2$ then there is $M_3$ such that $M_1 \reduces^* M_3$ and
$M_2 \reduces^* M_3$.
\end{thm}

\begin{proof}
By adapting Tait and Martin-L\"of's proof of Church-Rosser property of $\beta$-reduction of $\lambda$-calculus, 
which is based on the so called $1$-reduction 
$\twoheadrightarrow_1$. To the clauses for ordinary $\lambda$-calculus (see e.g.~\cite{Barendregt84} Definition 3.2.3) we add:
\begin{enumerate}[label=\roman{enumi})]
\item $\forall i \in I.  \ M_i \twoheadrightarrow_1 M'_i \Then \record{l_i=M_i \mid i\in I} \twoheadrightarrow_1
		 \record{l_i=M'_i \mid i\in I}$
\item $M \twoheadrightarrow_1 M' \Then M.l  \twoheadrightarrow_1 M'.l$
\item $M \twoheadrightarrow_1 M' \And R \twoheadrightarrow_1 R' \Then M\Override R \twoheadrightarrow_1 M'\Override R'$
\item $\forall i \in I.  \ M_i \twoheadrightarrow_1 M'_i  \And j\in I \Then \record{l_i=M_i \mid i\in I}.l_j \twoheadrightarrow_1 M'_j$
\item $\forall i \in I.  \ M_i \twoheadrightarrow_1 M'_i \And \forall j\in J. \ N_j \twoheadrightarrow_1 N'_j \Then$ \\
		$\Pair{l_i=M_i  \mid  i\in I} \Override \Pair{l_j=N_j  \mid \ j\in J} \twoheadrightarrow_1
		          \Pair{l_i=M'_i, \ l_j=N'_j \mid  i\in I\setminus J, \  j\in J}$
\end{enumerate}
Now it is easy to check that $\twoheadrightarrow_1^* \  = \ \reduces^*$ so that by Lemma 3.2.2 of \cite{Barendregt84} it suffices to
show that $\twoheadrightarrow_1$ satisfies the diamond property:
\[M \twoheadrightarrow_1 M_1 \And M\twoheadrightarrow_1 M_2 \Then \exists M_3. \ M_1\twoheadrightarrow_1 M_3 \And
M_2 \twoheadrightarrow_1 M_3.\]
The latter is easily established by induction over the definition of $M \twoheadrightarrow_1 M_1$ and by cases of
$M\twoheadrightarrow_1 M_2$.
\end{proof}

\medskip

In the spirit of Curry's assignment of polymorphic  types and of intersection types in particular, 
types are introduced as a syntactical tool to capture semantic properties of terms, rather than as constraints to term formation.

\begin{defi}[Intersection types for $\lambdaR$]  \label{def:intertype:lambdaR}
\[\begin{array}{rlll}
\TT \ni \sigma, \sigma_i & ::= &  a \mid \omega \mid \sigma_1\to\sigma_2 \mid \sigma_1 \inter \sigma_2 \mid  \rho & \text{types} \\ [1mm]
\TTR \ni \rho, \rho_i & ::= &   \record{} \mid \record{l:\sigma} \mid \rho_1 + \rho_2 \mid \rho_1 \inter \rho_2 & \text{record types}
\end{array}\]
where 
$a$ ranges over {\em type constants} and $l$ ranges over the denumerable set $\Label$.
\end{defi}
We use $\sigma,\tau$, possibly with sub and superscripts, for types in $\TT$ and $\rho,\rho_i$, possibly with superscripts, for record types in $\TTR$ only. Note that $\to$ associates to the right, and $\inter$ binds stronger than $\to$.
As with intersection type systems for the $\lambda$-calculus, the intended meaning of types are sets, provided 
a set theoretic interpretation of type constants $a$.

Following \cite{bcd}, type semantics is given axiomatically by means of the subtyping relation $\leq$, that can be interpreted as subset inclusion: see e.g.~\cite{Mit96} \S 10.4 for extending such interpretation to record types. It is the least pre-order over $\TT$ satisfying Definition~\ref{def:typeInclusionInt} and Definition~\ref{def:typeInclusionRec-new}.

\begin{defi}[Type inclusion: arrow and intersection types]\label{def:typeInclusionInt}
~ \hfill

\begin{enumerate}
\item \label{def:typeInclusionInt-1}
	$\sigma \leq \omega$ and $\omega \leq \omega\to\omega$
\item \label{def:typeInclusionInt-2}
	$\sigma\inter\tau \leq \sigma$ and $\sigma\inter\tau \leq \tau$
\item \label{def:typeInclusionInt-3}
	$\sigma \leq \tau_1 \And \sigma\leq \tau_2 \Then \sigma \leq \tau_1\inter\tau_2$
\item \label{def:typeInclusionInt-4}
	$(\sigma\to\tau_1) \inter (\sigma\to\tau_2) \leq \sigma \to \tau_1\inter\tau_2$
\item \label{def:typeInclusionInt-5}
	$\sigma_2 \leq \sigma_1 \And \tau_1 \leq \tau_2 \Then \sigma_1\to\tau_1 \leq \sigma_2\to\tau_2$
\end{enumerate}
\end{defi}

By \ref{def:typeInclusionInt}.\ref{def:typeInclusionInt-2} type $\omega$ is the top w.r.t. $\leq$; by \ref{def:typeInclusionInt}.\ref{def:typeInclusionInt-2}  and
\ref{def:typeInclusionInt}.\ref{def:typeInclusionInt-3} $\sigma \inter \tau$ is the meet. 
Writing $\sigma = \tau$ if $\sigma\leq \tau$ and $\tau\leq\sigma$, from
\ref{def:typeInclusionInt}.\ref{def:typeInclusionInt-1} and \ref{def:typeInclusionInt}.\ref{def:typeInclusionInt-5} we have
$\omega \leq \omega\to\omega \leq \sigma\to\omega \leq \omega$,
which are all equalities.
Finally from \ref{def:typeInclusionInt}.\ref{def:typeInclusionInt-4}  and \ref{def:typeInclusionInt}.\ref{def:typeInclusionInt-5} we also
have $(\sigma\to\tau_1) \inter (\sigma\to\tau_2) = \sigma \to \tau_1\inter\tau_2$.

If $\sigma_1 \leq \tau_1$ and $\sigma_2 \leq \tau_2$ then by transitivity we have 
$\sigma_1 \inter \sigma_2 \leq \sigma_1 \leq \tau_1$
and $\sigma_1 \inter \sigma_2 \leq \sigma_2 \leq \tau_2$ from which
$\sigma_1 \inter \sigma_2 \leq \tau_1 \inter \tau_2$ follows by \ref{def:typeInclusionInt}.\ref{def:typeInclusionInt-3}. Hence $\leq$ is
a precongruence w.r.t. $\inter$, and hence $=$ is a congruence.


\begin{defi}[Type inclusion: record types]\label{def:typeInclusionRec-new}
~ \hfill
\begin{enumerate}
\item \label{leq-ax-1} $\record{l:\sigma} \leq \record{}$
\item \label{leq-ax-2} $\Pair{l:\sigma}\inter\Pair{l:\tau} \leq \Pair{l:\sigma\inter\tau}$
\item \label{leq-ax-3} $\sigma \leq \tau \Then \Pair{l:\sigma} \leq \Pair{l:\tau}$
\item \label{leq-ax-4} $\rho + \record{} = \record{} + \rho = \rho$
\item \label{leq-ax-5} $(\rho_1 + \rho_2) + \rho_3 = \rho_1 + (\rho_2 + \rho_3)$
\item \label{leq-ax-6} $(\rho_1 \inter \rho_2) + \rho_3 = (\rho_1 + \rho_3) \inter (\rho_2 + \rho_3)$
\item \label{leq-ax-8} $\Pair{l:\sigma} + (\Pair{l:\tau} \inter \rho) = \Pair{l:\tau} \inter \rho$
\item \label{leq-ax-10} $\Pair{l:\sigma} + (\Pair{l':\tau} \inter \rho) = \Pair{l':\tau} \inter (\Pair{l:\sigma} + \rho)$ 
				if  $l\neq l'$
\item \label{leq-ax-11} $\rho_1 \leq \rho_2 \Then \rho_1 + \rho \leq \rho_2 + \rho$
\item \label{leq-ax-12} $\rho_1 = \rho_2 \Then \rho + \rho_1 = \rho + \rho_2$
\end{enumerate}
\end{defi}

While Definition \ref{def:typeInclusionInt} is standard after \cite{bcd}, comments on Definition \ref{def:typeInclusionRec-new} are in order.
First observe that from $\record{l:\sigma} \leq \record{}$ we obtain that $\record{l:\sigma} \inter \record{} =\record{l:\sigma}$, so that by putting
$\rho = \record{}$ in \ref{def:typeInclusionRec-new}.\ref{leq-ax-8} and in \ref{def:typeInclusionRec-new}.\ref{leq-ax-10} we obtain
by \ref{def:typeInclusionRec-new}.\ref{leq-ax-12}:
\begin{equation}\label{eq:interPlus}
\Pair{l:\sigma} + \Pair{l:\tau} =  \Pair{l:\tau}, \quad\quad
\Pair{l:\sigma} + \Pair{l':\tau} =  \Pair{l:\sigma} \inter \Pair{l':\tau} \qquad (l\neq l').
\end{equation}


Type $\record{}$ is the type of all records. 
Type $\Pair{l:\sigma}$ is a unary record type, whose meaning is the set of records having at least a field labeled by 
$l$, with value of type $\sigma$; therefore
$\record{l:\sigma} \inter \record{l:\tau}$ is the type of records having label $l$ with values both of type $\sigma$ and $\tau$, that is of type $\sigma\inter\tau$. 
In fact the following equation is derivable:
\begin{equation}\label{eq:interRecord}
\record{l:\sigma} \inter \record{l:\tau} = \Pair{l:\sigma\inter\tau}.
\end{equation}
 On the other hand
 $\record{l:\sigma} \inter \record{l':\tau}$, with $l\neq l'$, is the type of records having fields labeled by $l$ and $l'$, with values of type $\sigma$ and 
 $\tau$ respectively. It follows that intersection of record types can be used to express properties of records with arbitrary (though finitely) many fields, which justifies the abbreviations
$
\Pair{l_i:\sigma_i \mid i\in I \neq \emptyset} = \bigcap_{i\in I }  \Pair{l_i:\sigma_i}
$ and $\Pair{l_i:\sigma_i \mid i\in \emptyset} = \record{}$,
where we assume that the $l_i$ are pairwise distinct. 
Finally, as it will be apparent from Definition \ref{def:typeAssignment} below, $\rho_1+ \rho_2$ is the type of all records obtained by merging 
a record of type $\rho_1$ 
with a record of type $\rho_2$, which is intended to type $\Override$ that
is at the same time a record extension and field updating operation. Since this is the distinctive feature of the system introduced here, we comment on this by means of a few lemmas, illustrating its properties.

\begin{lem}\label{lem:derivedLeq}
~ \hfil
\begin{enumerate}
\item \label{lem:derivedLeq-b} $\Pair{l_{j_0}:\sigma} + \Pair{l_j : \tau_{j} \mid j \in J } =  \Pair{ l_j : \tau_{j} \mid j \in J }$ if $j_0 \in J$,
\item \label{lem:derivedLeq-c} $\Pair{l_{j_0}:\sigma} + \Pair{l_j : \tau_{j} \mid j \in J } = \Pair{l_{j_0}:\sigma}  \inter \Pair{ l_j : \tau_{j} \mid j \in J }$
	if $j_0 \not\in J$
\end{enumerate}
\end{lem}

\begin{proof} 
~ \hfill
\begin{enumerate}
\item By commutativity and associativity of the meet operator $\inter$ and by
	\ref{def:typeInclusionRec-new}.\ref{leq-ax-12}, when $j_0 \in J$ we may freely assume that
	\[ \Pair{l_{j_0}:\sigma} + \Pair{l_j : \tau_{j} \mid j \in J } = 
		\Pair{l_{j_0}:\sigma} + (\Pair{l_{j_0}:\tau_{j_0}} \inter \Pair{l_j : \tau_{j} \mid j \in J\setminus \Set{j_0} }) \]
	which is equal to $\Pair{l_{j_0}:\tau_{j_0}} \inter  \Pair{l_j : \tau_{j} \mid j \in J\setminus \Set{j_0} }$ 
	by \ref{def:typeInclusionRec-new}.\ref{leq-ax-8},
	namely to $\Pair{ l_j : \tau_{j} \mid j \in J }$.
\item By induction over the cardinality of $J$. If it is $0$ then $\Pair{l_j : \tau_{j} \mid j \in \emptyset} = \record{}$ and:
	\[  \Pair{l_{j_0}:\sigma} + \record{} = \Pair{l_{j_0}:\sigma} = \Pair{l_{j_0}:\sigma} \inter \record{} \]
	using \ref{def:typeInclusionRec-new}.\ref{leq-ax-4}, \ref{def:typeInclusionRec-new}.\ref{leq-ax-1} and the fact that $\inter$ is the meet.
	If $|J| > 0$ then, by reasoning as before we have:
	\[ \Pair{l_{j_0}:\sigma} + \Pair{l_j : \tau_{j} \mid j \in J } = 
		\Pair{l_{j_0}:\sigma} + (\Pair{l_{j_1}:\tau_{j_1}} \inter \Pair{l_j : \tau_{j} \mid j \in J\setminus \Set{j_1} }) .\]
	for some $j_1 \in J$. Now $j_0 \not\in J$ implies $l_{j_0} \neq l_{j_1}$, and the thesis follows 
	by \ref{def:typeInclusionRec-new}.\ref{leq-ax-10} 
	and the induction hypothesis.\qedhere
\end{enumerate}

\end{proof}

\begin{lem}\label{lem:normal-recordTypes}
~\hfill
\begin{enumerate}
\item \label{lem:normal-recordTypes-a}
	$\Pair{ l_i : \sigma_{i} \mid i \in I } \leq \Pair{l_j : \tau_{j} \mid j \in J } \Iff J \subseteq I \And \forall j \in J.\;\sigma_j \leq \tau_j $,
\item \label{lem:normal-recordTypes-ab}
	$ \Pair{ l_i : \sigma_{i} \ \mid \ i \in I } \inter  \Pair{ l_j : \tau_{j} \mid j \in J } =
		\Pair{l_i : \sigma_{i},  l_j : \tau_{j}, l_k:\sigma_k\inter\tau_k \mid i\in I\setminus J, j\in J\setminus I, k \in I \cap J}$,
\item \label{lem:normal-recordTypes-b}
	$ \Pair{ l_i : \sigma_{i} \ \mid \ i \in I } +  \Pair{ l_j : \tau_{j} \mid j \in J } =
		\Pair{l_i : \sigma_{i},  l_j : \tau_{j} \mid i\in I\setminus J, j\in J}$,
\item \label{lem:normal-recordTypes-c}
	$\forall \rho \in \TTR .\; \exists \ \Pair{l_i:\sigma_i  \mid  i\in I}.~
		\rho = \Pair{l_i:\sigma_i \mid  i\in I} $.
\end{enumerate}
\end{lem}

\begin{proof} Through this proof let $\rho_1 \equiv \Pair{ l_i : \sigma_{i} \mid i \in I }$ and $\rho_2 \equiv \Pair{l_j : \tau_{j} \mid j \in J }$,
where $\equiv$ denotes syntactic identity up to commutativity and associativity of $\inter$.

\begin{enumerate}

\item The only if part is proved by induction over the derivation of $\rho_1 \leq \rho_2$. If this equation is an instance of axiom \ref{def:typeInclusionRec-new}.\ref{leq-ax-1} then $J = \emptyset \subseteq I$ and $\forall j \in J.\;\sigma_j \leq \tau_j$ is vacuously true.

\noindent
Axiom \ref{def:typeInclusionRec-new}.\ref{leq-ax-2} doesn't apply being the labels $l_i$ pairwise distinct. 

\noindent
In case the derivation ends by rule \ref{def:typeInclusionRec-new}.\ref{leq-ax-3} we have $I = J = \Set{l}$ and $\sigma \leq \tau$ by the 
premise of the rule.

\noindent
In case of the axiom instance $\record{l_1:\sigma} \inter \record{l_2:\tau} \leq \record{l_1:\sigma}$ of 
\ref{def:typeInclusionInt}.\ref{def:typeInclusionInt-2} we have
$J = \Set{l_1} \subseteq \Set{l_1, l_2} = I$ and obviously $\sigma \leq \sigma$. 

\noindent
Finally in case of deriving $\rho_1 \leq (\rho_3 \inter \rho_4)$ from $\rho_1 \leq \rho_3$ and 
$\rho_1 \leq \rho_4$, where $\rho_2 \equiv \rho_3 \inter \rho_4$, we have that for some $J', J''$ such that $J = J' \cup J''$ it
is $\rho_3 \equiv \Pair{l_h : \tau_{h} \mid h \in J'}$ and $\rho_4 \equiv \Pair{l_k : \tau_{k} \mid k \in J''}$. By induction we know that
$J' \subseteq I$ and $\sigma_h \leq \tau_h$ for all $h\in J'$ and $J'' \subseteq I$ and $\sigma_k \leq \tau_k$ for all $k\in J''$. Therefore
$J = J' \cup J'' \subseteq I$ and $\sigma_j \leq \tau_j$ for all $j \in J$.

\noindent
For the if part we reason by induction over the cardinality of $J$. If $|J| = 0$ then $\rho_2 \equiv \record{}$, hence either $I = \emptyset$, but
then $\rho_1 \equiv \record{}$; or $I \neq \emptyset$ so that
$\rho_1 \leq \Pair{l_i:\sigma_i} \leq \record{}$ for any $i \in I$, using axiom \ref{def:typeInclusionRec-new}.\ref{leq-ax-1}.
If $|J| > 0$ then let $j_0 \in J$ be any index: then $\rho_2 \equiv \Pair{l_j:\tau_j \mid j \in J\setminus \Set{j_0}} \inter \Pair{l_{j_0}:\tau_{j_0}}$.
Now by induction $\rho_1 \leq \Pair{l_j:\tau_j \mid j \in J\setminus \Set{j_0}}$; by hypothesis $\sigma_{j_0}\leq \tau_{j_0}$  so
that $\rho_1 \leq  \Pair{l_{j_0}:\sigma_{j_0}} \leq  \Pair{l_{j_0}:\tau_{j_0}}$ and we conclude being $\inter$ the meet w.r.t. $\leq$. 

\item
This is an immediate consequence of commutativity, associativity, idempotency of $\inter$ and 
equation (\ref{eq:interRecord}).

\item
If $I = \emptyset$ then$\rho_1  = \record{}$ and the thesis follows by 
\ref{def:typeInclusionRec-new}.\ref{leq-ax-4}. Otherwise we have:
\[\begin{array}{llll}
\rho_1 + \rho_2 & = & \bigcap_{i\in I} (\record{l_i:\sigma_i} + \rho_2) &
	 \mbox{by repeated applications of \ref{def:typeInclusionRec-new}.\ref{leq-ax-6}} \\
& = & \bigcap_{i\in I} \Pair{l_i:\sigma_i , l_j: \tau_j\mid i \not \in J, j \in J} & 
	\mbox{by Lemma \ref{lem:derivedLeq}.\ref{lem:derivedLeq-b} and \ref{lem:derivedLeq}.\ref{lem:derivedLeq-c}} \\
& \equiv & \Pair{l_i : \sigma_{i},  l_j : \tau_{j} \mid i\in I\setminus J, j\in J} .
\end{array}\]
	
\item By induction over $\rho$. If $\rho = \record{}$ there is nothing to prove. 
	
	\noindent
	If $\rho \equiv \rho_1 \inter \rho_2$ then by induction 
	$\rho_1 = \Pair{l_i^1:\sigma_i^1 \mid i \in I_1}$ and
	$\rho_2 = \Pair{l_i^2:\sigma_i^2 \mid i \in I_2}$ so that the thesis follows by
	(\ref{lem:normal-recordTypes-ab}) above and the fact that $=$ is a congruence w.r.t. $\inter$.
	
	\noindent
	Similarly if $\rho \equiv \rho_1 + \rho_2$ the thesis follows by induction using
	(\ref{lem:normal-recordTypes-b}) and the fact that $=$ is a congruence w.r.t. $+$
	by \ref{def:typeInclusionRec-new}.\ref{leq-ax-11} and \ref{def:typeInclusionRec-new}.\ref{leq-ax-12}.\qedhere
\end{enumerate}
\end{proof}

\begin{rem}\label{rem:interPlus}
By \ref{def:typeInclusionRec-new}.\ref{leq-ax-6} the $+$ distributes to the left over $\inter$;
however this doesn't hold to the right, namely
$\rho_1+(\rho_2\inter\rho_3) \neq (\rho_1+\rho_2)\inter(\rho_1+\rho_3)$: take $\rho_1\equiv \record{l_1:\sigma_1, l_2:\sigma_2}$,
$\rho_2 \equiv \record{l_1:\sigma'_1}$ and $\rho_3 \equiv \record{l_2:\sigma'_2}$, with $\sigma_1 \neq \sigma'_1$ and $\sigma_2 \neq \sigma'_2$. Then by Lemma \ref{lem:normal-recordTypes}.\ref{lem:normal-recordTypes-b} we have:
$\rho_1+(\rho_2\inter\rho_3) = \rho_1 + \record{l_1:\sigma'_1, l_2:\sigma'_2} = \record{l_1:\sigma'_1, l_2:\sigma'_2}$,
while
$ (\rho_1+\rho_3)\inter(\rho_2+\rho_3) = \record{l_1:\sigma_1, l_2:\sigma'_2} \inter\record{l_1:\sigma'_1, l_2:\sigma_2} =
\record{l_1:\sigma_1 \inter \sigma'_1, l_2:\sigma_2 \inter\sigma'_2}.
$
The last example suggests that $(\rho_1+\rho_2)\inter(\rho_1+\rho_3) \leq \rho_1+(\rho_2\inter\rho_3)$.

On the other hand by \ref{def:typeInclusionRec-new}.\ref{leq-ax-11} $+$ is monotonic in its first argument. However we have $\rho_2\leq\rho_3 \not\Rightarrow \rho_1+\rho_2 \leq \rho_1+\rho_3$. Indeed: 
\[
\begin{array}{llll}
\record{l_0:\sigma_1,l_1:\sigma_2} + \record{l_1:\sigma_3, l_2:\sigma_4} & = & \record{l_0:\sigma_1,l_1:\sigma_3, l_2:\sigma_4} \\
& \not\leq &  \record{l_0:\sigma_0,l_1:\sigma_1, l_2:\sigma_4} & \mbox{if $\sigma_3\not\leq\sigma_1$}\\
& = & \record{l_0:\sigma_1,l_1:\sigma_2} + \record{l_2:\sigma_4}
\end{array}\]
even if $\record{l_1:\sigma_1, l_2:\sigma_2} \leq \record{l_2:\sigma_2}$.
From this we conclude that $+$ is not monotonic in its second argument. Comparing this with \ref{def:typeInclusionRec-new}.\ref{leq-ax-12} we
see that $\leq$ is not a precongruence w.r.t. $+$, while $=$, namely the symmetric closure of $\leq$, is a congruence.

Finally if one assumes (\ref{lem:normal-recordTypes-a})-(\ref{lem:normal-recordTypes-b}) of Lemma \ref{lem:normal-recordTypes} as axioms
then (pre)-congruence axioms in Definition \ref{def:typeInclusionRec-new} become derivable. 
The opposite is hardly provable and possibly false\footnote{Remark of an anonymous referee about our former attempt to derive
\ref{def:typeInclusionRec-new}.\ref{leq-ax-11} and \ref{def:typeInclusionRec-new}.\ref{leq-ax-12} from the other axioms.}.
\end{rem}

Lemma \ref{lem:normal-recordTypes}.\ref{lem:normal-recordTypes-a} states that subtyping among intersection of unary record types
subsumes subtyping in width and depth of ordinary record types from the literature. 
Lemma \ref{lem:normal-recordTypes}.\ref{lem:normal-recordTypes-b} shows that
the $+$ type constructor reflects at the level of types the operational behavior of the merge operator $\oplus$.
Lemma \ref{lem:normal-recordTypes}.\ref{lem:normal-recordTypes-c} says that any record type is equivalent to an intersection of unary record types; this
implies that types of the form $\rho_1 + \rho_2$ are eliminable in principle. 
However, they play a key role in typing mixins, motivating the issue of control of negative information in the synthesis process: 
see Section \ref{sec:synthesis}.
More properties of subtyping record types w.r.t. $+$ and $\inter$ are
listed in the next two lemmas. 

\begin{lem}\label{lem:typePlus}
$\rho_1 + \rho_2 = \rho_1\inter\rho_2 \Iff  \rho_1+\rho_2 \leq \rho_1$.
\end{lem}

\begin{proof} W.l.o.g. by \ref{lem:normal-recordTypes}.\ref{lem:normal-recordTypes-c} 
let us assume that $\rho_1 = \record{l_i:\sigma_{1,i} \mid i\in I}$, $\rho_2 = \record{l_j:\sigma_{2,j} \mid j\in J}$ and
$\rho_3 = \record{l_k:\sigma_{3,k} \mid k\in K}$. 
If $\rho_1 + \rho_2 = \rho_1\inter\rho_2$ then by 
 \ref{lem:normal-recordTypes}.\ref{lem:normal-recordTypes-ab} and 
  \ref{lem:normal-recordTypes}.\ref{lem:normal-recordTypes-b}:
\begin{align*}
&\record{l_i: \sigma_{1,i}, l_j:\sigma_{2,j} \mid i \in I\setminus J, j\in J}\\
=\;&\record{l_i: \sigma_{1,i}, l_j:\sigma_{2,j}, l_k:\sigma_{1,k} \inter\sigma_{2,k}   \mid i \in I\setminus J, j\in J\setminus I, k\in I \cap I}
\end{align*}
Obviously $J = (J\setminus I) \cup (I\cap J)$ and from \ref{lem:normal-recordTypes}.\ref{lem:normal-recordTypes-a}
we deduce that $\sigma_{2,j} = \sigma_{1,j} \inter\sigma_{2,j}$ for all $j\in J$,
hence $\sigma_{2,j} \leq \sigma_{1,j}$ so we conclude that $\rho_1+\rho_2 \leq \rho_1$ again by 
\ref{lem:normal-recordTypes}.\ref{lem:normal-recordTypes-a}. Since all these implications can be reverted, we conclude.
\end{proof}

Let us define the map 
$\lbl: \TTR \to \wp(\Label)$ (where $\wp(\Label)$ is the powerset of $\Label$) by:
\[
\begin{array}{llll}
\lbl(\record{l:\sigma}) = \Set{l}, & \lbl(\rho_1 \inter \rho_2) = \lbl(\rho_1 + \rho_2) = \lbl(\rho_1) \cup \lbl(\rho_2).
\end{array}
\]

\begin{lem}\label{lem:label}
~\hfill
\begin{enumerate}
\item \label{lem:label-a}
	$\rho_1 = \rho_2 \Then \lbl(\rho_1) = \lbl(\rho_2)$,
\item \label{lem:label-b}
	$\lbl(\rho_1) \cap \lbl(\rho_2) = \emptyset \Then \rho_1 + \rho_2 = \rho_1 \inter \rho_2$,
\item \label{lem:label-bb}
	$\lbl(\rho_1) \subseteq \lbl(\rho_2)  \Then \rho_1 + \rho_2 = \rho_2$,
\item \label{lem:label-c}
	$\lbl(\rho_2) = \lbl(\rho_3) \Then (\rho_1 + \rho_2) \inter (\rho_1 + \rho_3) = \rho_1 + (\rho_2 \inter \rho_3)$.
\end{enumerate}
\end{lem}

\begin{ProofNoQed} Part (\ref{lem:label-a}) follows by \ref{lem:normal-recordTypes}.\ref{lem:normal-recordTypes-a}, and
(\ref{lem:label-b}) is a consequence of the fact that if $\lbl(\rho_1) \cap \lbl(\rho_2) = \emptyset$ then equation
(\ref{eq:interPlus}) $\Pair{l:\sigma} + \Pair{l':\tau} =  \Pair{l:\sigma} \inter \Pair{l':\tau}$
repeatedly applies. Part (\ref{lem:label-bb}) follows by
 \ref{lem:normal-recordTypes}.\ref{lem:normal-recordTypes-b}.
 
To see part (\ref{lem:label-c}) we can assume by \ref{lem:normal-recordTypes}.\ref{lem:normal-recordTypes-c} that
$\rho_1 = \record{l_i:\sigma_{1,i} \mid i\in I}$, $\rho_2 = \record{l_j:\sigma_{2,j} \mid j\in J}$ and
$\rho_3 = \record{l_k:\sigma_{3,k} \mid k\in K}$. Since $\lbl(\rho_2) = \lbl(\rho_3)$ we have that $J = K$, so that:
\[\begin{array}{llll}
\rho_1 + \rho_2 & = & \record{l_i:\sigma_{1,i}, l_j:\sigma_{2,j}  \mid i\in I\setminus J, j\in J} & 
	\mbox{by \ref{lem:normal-recordTypes}.\ref{lem:normal-recordTypes-b}}\\ [1mm]
\rho_1 + \rho_3 & = & \record{l_i:\sigma_{1,i}, l_j:\sigma_{3,j}  \mid i\in I\setminus J, j\in J} & 
	\mbox{by \ref{lem:normal-recordTypes}.\ref{lem:normal-recordTypes-b}}\\ [1mm]
\rho_2 \inter \rho_3 & = & \record{l_j:\sigma_{2,j}\inter\sigma_{3,j}  \mid j \in J } & 
	\mbox{by \ref{lem:normal-recordTypes}.\ref{lem:normal-recordTypes-ab}}
\end{array}\]
Again by \ref{lem:normal-recordTypes}.\ref{lem:normal-recordTypes-ab}, \ref{lem:normal-recordTypes}.\ref{lem:normal-recordTypes-b}
and the fact that $\sigma_{1,i} \inter \sigma_{1,i} = \sigma_{1,i}$ 
we conclude that $(\rho_1 + \rho_2) \inter (\rho_1 + \rho_3)$ and $\rho_1 + (\rho_2 \inter \rho_3)$ are both equal to
\[
  \record{l_i:\sigma_{1,i}, l_j:\sigma_{2,j}\inter\sigma_{3,j} \mid  i\in I\setminus J, j\in J}.\tag*{\qEd}
\]
\end{ProofNoQed}

About Lemma \ref{lem:label}.\ref{lem:label-b} above note that 
condition $\lbl(\rho_1) \cap \lbl(\rho_2) = \emptyset$ is essential, since $\inter$ is commutative while
$\rho_1 + \rho_2 \neq \rho_2 + \rho_1$ in general, as it immediately follows by 
Lemma \ref{lem:normal-recordTypes}.\ref{lem:normal-recordTypes-b}.

%
%

\medskip
We come now to the type assignment system. A {\em basis} (also called a context in the literature) 
is a finite set $\Context = \Set{x_1:\sigma_n,\ldots, x_n:\sigma_n}$, where the variables $x_i$ are pairwise distinct; we set
$\dom(\Context) = \Set{x \mid \exists \, \sigma.~x:\sigma \in\Context}$ and we write $\Context, x:\sigma$ for $\Context\cup\Set{x:\sigma}$
where $x\not\in\dom(\Context)$. Then we consider the following extension of the
system in \cite{bcd}, also called {\bf BCD} in the literature.

\begin{defi}[Type Assignment]\label{def:typeAssignment}
The rules of the assignment system are:
\[
\scalebox{0.9}{\mbox{\ensuremath{\displaystyle{\begin{array}{ccc }
\prooftree
	x:\sigma \in \Context
	\justifies
	\Context \deduces x:\sigma
	\using (\Axiom)
\endprooftree 
&
\prooftree
	\Context, x:\sigma \deduces M: \tau
	\justifies
	\Context \deduces \lambda x.M: \sigma\to\tau
	\using (\ArrI)
\endprooftree 
&
\prooftree
	\Context \deduces M: \sigma\to\tau \quad \Context \deduces N:\sigma
	\justifies
	\Context \deduces MN:\tau
	\using (\ArrE)
\endprooftree 

\\ [6mm]

\prooftree
	\Context \deduces M: \sigma \quad 
	\Context \deduces M: \tau
	\justifies
	\Context \deduces M: \sigma\inter\tau
	\using (\inter)
\endprooftree 
&
\prooftree
	\justifies
	\Context \deduces M:\omega
	\using (\omega)
\endprooftree 
&
\prooftree
	\Context \deduces M: \sigma \qquad \sigma\leq\tau
	\justifies
	\Context \deduces M : \tau
	\using (\leq)
\endprooftree

\\ [6mm]

\prooftree
	\justifies
	\Context \deduces \record{ l_i = M_i \ | \ i \in I } : \record{}
	\using (\record{})
\endprooftree 

&

\prooftree
	\Context \deduces M_k: \sigma \quad 
	k \in I
	\justifies
	\Context \deduces \record{ l_i = M_i \ | \ i \in I } : \Pair{l_k: \sigma}
	\using (\recRule)
\endprooftree 
&
\prooftree
	\Context \deduces M: \Pair{l:\sigma}
	\justifies
	\Context \deduces M.l : \sigma
	\using (\selRule)
\endprooftree 

\\ [6mm]

\multicolumn{3}{c}{
\prooftree
	\Context \deduces M: \rho_1 \qquad \Context \deduces R: \rho_2 \quad (*)
	\justifies
	\Context \deduces M \oplus R : \rho_1 + \rho_2
	\using (+)
\endprooftree 
}

\end{array}}}}}\]
where $(*)$ in rule $(+)$ is the side condition: $\lbl(R) = \lbl(\rho_2)$. 
\end{defi}

\noindent
Using Lemma \ref{lem:normal-recordTypes}.\ref{lem:normal-recordTypes-a}, the following rule is easily shown to be admissible:
\[
\prooftree
	\Context \deduces M_j:\sigma_j \qquad \forall j \in J \subseteq I
	\justifies
	\Context \deduces \record{l_i=M_i \mid  i\in I}: \Pair{l_j:\sigma_i \mid  j\in J} 
\using (\recRule')
\endprooftree 
\]
Contrary to this, the side condition $(*)$ of rule $(+)$ is equivalent to ``exact'' record typing in \cite{BrachaC90}, disallowing record subtyping in width.
Such a condition is necessary for soundness of typing.
Indeed, suppose that $\Context\der M_0:\sigma$ and $\Context\der M'_0:\sigma'_0$ but $\Context \not\der M'_0:\sigma_0$; then without $(*)$
we could derive:
\[
\prooftree
	\Context \der \record{l_0=M_0} : \Pair{l_0:\sigma_0} \qquad 
	\Context \der \record{l_0=M'_0,l_1:\sigma_1} : \Pair{l_1:\sigma_1}
\justifies
	\Context \der \record{l_0=M_0}\Override \record{l_0=M'_0,l:_1:\sigma_1} :  \record{l_0:\sigma_0, l_1:\sigma_1}
\endprooftree
\]
from which we obtain that $\Context \der (\record{l_0=M_0}\Override \record{l_0=M'_0,l:_1:\sigma_1}).l_0:\sigma_0$ breaking subject reduction, since
$(\record{l_0=M_0}\Override \record{l_0=M'_0,l:_1:\sigma_1}).l_0 \reduces^* M'_0$. The essential point is that proving that
$\Context\der N: \record{l:\sigma}$ doesn't imply that $l'\not\in\lbl(R')$ for any $l'\neq l$, which follows only by the uncomputable (not even recursively enumerable)
statement that $\Context\not\der N: \record{l':\omega}$, a negative information.

This explains the restriction to record terms as the second argument of $\Override$: 
in fact, allowing $M\Override N$ to be well formed for an arbitrary $N$ we might have $N\equiv x$ in $\lambda x.\,(M\Override x)$.
But extending $\lbl$ to all terms in $\lambdaR$ is not possible without severely limiting the expressiveness of the assignment system. In fact to say that
$\lbl(N) = \lbl(R)$ if $N\reduces^* R$ would make the $\lbl$ function non computable; on the other hand putting $\lbl(x)=\emptyset$, which is the only 
reasonable and conservative choice as we do not know possible substitutions for $x$ in $\lambda x.\,(M\Override x)$, implies that the latter term
has type $\omega\to\omega = \omega$ at best.

As a final remark, let us observe that we do not adopt exact typing of records in general, but only for typing the right-hand side of $\oplus$-terms, a feature that will be essential when typing mixins.

The following two lemmas are standard after \cite{bcd}.

\begin{lem}\label{lem:arrow-prop}
\[ \bigcap_{i\in I} (\sigma_i \to \tau_i) \leq \sigma \to \tau \Then
	\exists J \subseteq I.\;  \sigma \leq \bigcap_{j\in J} \sigma_j \And 
	      \bigcap_{j\in J}\tau_j \leq \tau \]
\end{lem}

\begin{proof}
By induction over the proof of $\bigcap_{i\in I} (\sigma_i \to \tau_i) \leq \sigma \to \tau$.
\end{proof}

\begin{lem}\label{lem:weak-leq}
The following rule is admissible:
\[
\prooftree
	\Context, x:\tau \deduces M:\sigma \quad \tau' \leq \tau
\justifies
	\Context, x:\tau' \deduces M:\sigma
\endprooftree
\]
\end{lem}

\newenvironment{bprooftree}
 {\begin{adjustbox}{raise=\depth}\begin{prooftree}}
 {\end{prooftree}\end{adjustbox}}
\begin{ProofNoQed}
 By induction over the derivation of $\Context, x:\tau \deduces M:\sigma$. 
The only non-trivial case is when  $\Context, x:\tau \deduces M:\sigma$ is an instance of $(\Axiom)$ and $M \equiv x$.
In this case $\sigma = \tau$ and we replace the axiom by the inference:
\[
\begin{bprooftree} \Context, x:\tau' \deduces x:\tau' \quad \tau'\leq \tau \justifies
  \Context, x:\tau' \deduces x:\tau \using (\leq)
\end{bprooftree}
\tag*{\qEd}
\]
\end{ProofNoQed}

\begin{lem}[Generation]\label{lem:generation}
Let $\sigma \neq \omega$:
\begin{enumerate}

\item\label{gen:var} 
	$\Context \deduces x: \sigma  \Iff  \exists \tau.\; x:\tau \in \Context \And \tau \leq \sigma$,

\item\label{gen:fun} 
	$\Context \deduces \lambda x.M : \sigma \Iff  
	\exists \; I, (\sigma_i)_{i \in I}, (\tau_i)_{i \in I}.~~\Context, x: \sigma_i \deduces M: \tau_i \And
	\bigcap_{i\in I}  (\sigma_i \arrow \tau_i) \leq \sigma$,

\item\label{gen:app} 
	$\Context \deduces MN : \sigma  \Iff  \exists \; \tau.~
	\Context \deduces M:\tau\to\sigma \And \Context \deduces N:\tau$,
	
\item\label{gen:rec} 
	$\Context \deduces \record{l_i=M_i \mid  i\in I}: \sigma  \Iff
	\forall i \in I \ \exists \; \sigma_i.\; \Context \deduces M_i:\sigma_i \And \record{l_i:\sigma_i \mid i\in I} \leq \sigma$,

\item\label{gen:sel} 
	$\Context \deduces M.l: \sigma  \Iff \Context \deduces M: \record{l:\sigma}$,

\item\label{gen:merge}
	 $\Context \deduces M \oplus R : \sigma  \Iff 
	 \exists \: \rho_1,\rho_2.~ \Context \deduces M:\rho_1 \And \Context \deduces R: \rho_2 \And
	 \lbl(R)  = \lbl(\rho_2) \And \rho_1+\rho_2 \leq \sigma$.
\end{enumerate}
\end{lem}

\begin{ProofNoQed} All the if parts are obvious. For the only if parts
first observe that there is a one-to-one correspondence between the term constructors and the rules in the type system but in case of
rules $(\omega)$, $ (\inter)$ and $(\leq)$. Indeed, rule $(\record{})$ is no exception as the typing
$\Context \deduces \record{ l_i = M_i \ | \ i \in I } : \record{}$ is derivable using $(\recRule)$ and $(\leq)$ in all cases but when $I = \emptyset$, in
which case $(\record{})$ becomes the axiom $\Context\deduces \record{}:\record{}$ and $\record{}$ is a nullary operator.

Disregarding rule $(\omega)$ because of the hypothesis $\sigma \neq \omega$, all the statements in this lemma depend on the remark
that any derivation $\Der$ of $\Context \deduces M:\sigma$ consists of a finite set of subderivations  $\Der_i$ of judgments
$\Context \deduces M:\sigma_i$ such that each $\Der_i$  ends by the rule corresponding to the main term constructor of $M$, and that
in $\Der$ all the inferences below the $\Der_i$ are instances of either rule $(\inter)$ or $(\leq)$. By noting that inferences
using $(\inter)$ and $(\leq)$ commute that is:
{\small
\[
\prooftree
	\prooftree
		\Context \deduces M: \sigma' \quad \sigma'\leq\sigma
	\justifies
		\Context \deduces M: \sigma
	\using (\leq)
	\endprooftree
	\quad
	\Context \deduces M:\tau
\justifies
	\Context \deduces M:\sigma\inter\tau
\using (\inter)
\endprooftree
~~\mbox{becomes}~~
\prooftree
	\prooftree
		\Context \deduces M: \sigma' \quad \Context \deduces M:\tau
	\justifies
		\Context \deduces M: \sigma'\inter\tau
	\using (\inter)
	\endprooftree
	\quad
	\sigma'\inter\tau\leq\sigma\inter\tau
\justifies
	\Context \deduces M:\sigma\inter\tau
\using (\leq)
\endprooftree
\]
}
we can freely assume that all $(\inter)$ inferences precede $(\leq)$ rules, concluding that $\bigcap_i \sigma_i \leq \sigma$.

Given that all the statements of this lemma are deduced by reading backward rules $(\Axiom)$, $(\ArrI)$, $(\ArrE)$, $(\record{})$,
$(\recRule)$, $(\selRule)$ and $(+)$. All cases are either standard from \cite{bcd} or are easy extensions thereof, but that of rule
$(+)$ in $(\ref{gen:merge})$ above. In this case we know that in the derivation of $\Context \deduces M \oplus R : \sigma$ 
there are subderivations ending by the inference:
\[
\prooftree
	\Context \deduces M:\rho_{i,1} \quad \Context \deduces R:\rho_{i,2} \quad \lbl(\rho_{i,2}) = \lbl(R)
\justifies
	\Context \deduces M \oplus R : \rho_{i,1} + \rho_{i,2}
\using (+)
\endprooftree
\]
where the side condition  $\lbl(\rho_{i,2}) = \lbl(R)$ holds for all $i\in I$; then by the remark above
we have that $\bigcap_{i\in I}( \rho_{i,1} + \rho_{i,2}) \leq \sigma$. Now taking $\rho_1 = \bigcap_{i\in I}  \rho_{i,1}$ and 
$\rho_2 = \bigcap_{i\in I} \rho_{i,2}$ we have that:
\[
\begin{array}[b]{rcll}
\sigma & \geq & \bigcap_{i\in I}( \rho_{i,1} + \rho_{i,2}) \\
& \geq & \bigcap_{i\in I}( \rho_{1} + \rho_{i,2}) & 
		\mbox{by Def. \ref{def:typeInclusionRec-new}.\ref{leq-ax-11}, since $\rho_{i,1} \geq \rho_1$ for all $i\in I$} \\ 
& = & \rho_1 + \rho_2 & 
	\mbox{by \ref{lem:label}.\ref{lem:label-c} since $\lbl(\rho_2) = \lbl(\rho_{i,2})$ for all $i\in I$.}
    \tag*{\qEd}
\end{array}\]

\end{ProofNoQed}

\begin{lem}[Substitution]\label{lem:substitution}
\[ \Context, x:\sigma \deduces M:\tau \And \Context\deduces N:\sigma \Then \Context \deduces M\Subst{N}{x}:\tau. \]
\end{lem}

\begin{proof} By induction over the derivation of $\Context, x:\sigma \deduces M:\tau$. Observe that the writing
$\Context, x:\sigma$ implies that $x\not\in\Context$; on the other hand if $\Context\deduces N:\sigma$ is derivable then
$\fv(N) \subseteq \Context$, hence $x\not\in\fv(N)$.
\end{proof}

\begin{thm}[Subject reduction]\label{thm:subjectRed}
$\Context \deduces M:\sigma \And M \reduces N \Then \Context \deduces N:\sigma$.
\end{thm}


\begin{proof} 
The proof is by cases of reduction rules, using Lemma \ref{lem:generation}. 
\begin{enumerate}
\item Case $M \equiv  (\lambda x.M')N'$.
Then $N \equiv M' \Subst{N'}{x}$ by rule ($\beta$). 
By Lemma \ref{lem:generation}.\ref{gen:app}, 
there exists $\tau$ such that $\Gamma \deduces (\lambda x.M'): \tau \to \sigma$ and
$\Gamma \deduces N' : \tau$.
By Lemma \ref{lem:generation}.\ref{gen:fun}, 
there exist $I, \sigma_i, \tau_i$ such that 
$\Gamma, x: \tau_i \deduces M': \sigma_i$ for all $i\in I$ and 
$\bigcap_{i\in I}  (\tau_i \arrow \sigma_i) \leq \tau \to \sigma$. By Lemma \ref{lem:arrow-prop} 
this implies that there is $J \subseteq I$ such that $\tau \leq \bigcap_{j\in J}  \tau_j$ and
$\bigcap_{j\in J} \sigma_j\leq \sigma$.
By Lemma \ref{lem:weak-leq} this implies that $\Gamma, x: \tau \deduces M': \sigma_j$ for all $j\in J$, so that
by rule $(\inter)$ we have $\Gamma, x: \tau \deduces M': \bigcap_{j\in J} \sigma_j$ and hence
$\Gamma, x: \tau \deduces M': \sigma$ by $(\leq)$. From this and $\Gamma \deduces N' : \tau$ we get
$\Gamma \deduces M' \Subst{N'}{x}  : \sigma$ by Lemma \ref{lem:substitution}.

\item Case $M \equiv \Pair{l_i=M_i  \mid  i\in I}.l_j$.
Then $N \equiv M_j$ for some $j\in I$ by rule ($\redSel$).
By Lemma \ref{lem:generation}.\ref{gen:sel} 
$\Gamma \deduces M: \record{l_j:\sigma}$, hence by \ref{lem:generation}.\ref{gen:rec}
for all $i\in I$ there exists $\sigma_i$ such that $\Gamma \deduces M_i:\sigma_i$ and
$\record{l_i:\sigma_i \mid i\in I} \leq \record{l_j:\sigma}$. 
Now by the fact that $j\in I$ and by 
Lemma \ref{lem:normal-recordTypes}.\ref{lem:normal-recordTypes-a}  it follows that 
$\Gamma \deduces M_j:\sigma_j$ and $\sigma_j \leq \sigma$, hence
$\Gamma \deduces M_j:\sigma$ by rule $(\leq)$.

\item Case $M \equiv \Pair{l_i=M_i  \mid  i\in I} \Override \Pair{l_j=N_j  \mid \ j\in J}$.
Then by rule ($\redMergeThree$):
\[ N \equiv \Pair{l_i=M_i, \ l_j=N_j \mid  i\in I\setminus J, \  j\in J}.\]
By Lemma \ref{lem:generation}.\ref{gen:merge} there exist the record types $\rho_1,\rho_2$ such that:
\begin{enumerate}
\item \label{SR-a}
	$\Gamma \deduces \Pair{l_i=M_i  \mid  i\in I}:\rho_1$,
\item \label{SR-b}
	$\Gamma \deduces \Pair{l_j=N_j  \mid  j\in J}:\rho_2$,
\item \label{SR-c}
	$\lbl(\rho_2) = \lbl(\Pair{l_j=N_j  \mid \ j\in J}) = \Set{l_j \mid j\in J}$,
\item \label{SR-d}
	$\rho_1 + \rho_2 \leq \sigma$.
\end{enumerate}
By (\ref{SR-a}) and (\ref{SR-b}) and by Lemma \ref{lem:generation}.\ref{gen:rec} we have that
for all $i\in I$ there exist $\sigma_i$ such that $\Gamma \deduces M_i:\sigma_i$, with
$\record{l_i:\sigma_i \mid i\in I} \leq \rho_1$, and similarly for all $j\in J$ there are $\tau_j$ such that
$\Gamma \deduces N_j:\tau_j$ and $\record{l_j:\tau_j \mid j\in J} \leq \rho_2$.

Since $\Gamma \deduces M_i:\sigma_i$ for all $i\in I$, a fortiori it holds for all $i\in I\setminus J$. On the other hand,
by Lemma \ref{lem:normal-recordTypes}.\ref{lem:normal-recordTypes-c} we know that
$\rho_2 = \record{l_k:\tau'_k \mid k\in K}$ for some $K$ and $\tau'_k$, so that $\record{l_j:\tau_j \mid j\in J} \leq \rho_2$
implies $J \supseteq K$ and $\tau_k \leq \tau'_k$ for all $k\in K$ by Lemma  \ref{lem:normal-recordTypes}.\ref{lem:normal-recordTypes-a}.
It follows that $J = K$ by (\ref{SR-c}), and $\Gamma \deduces N_j:\tau'_j$ for all $j\in J$ by rule $(\leq)$. Then by rule $ (\recRule)$ and $(\inter)$
we conclude that
\[ \Gamma \deduces N: \record{l_i : \sigma_i, \ l_j : \tau'_j \mid  i\in I\setminus J, \  j\in J}. \]
Now $\record{l_i : \sigma_i, \ l_j : \tau'_j \mid  i\in I\setminus J, \  j\in J} = \record{l_i:\sigma_i \mid i\in I} + \rho_2$
by Lemma \ref{lem:normal-recordTypes}.\ref{lem:normal-recordTypes-b}, but since
\begin{align*}
  \record{l_i:\sigma_i \mid i\in I} \leq \rho_1
\end{align*}
we have by Definition \ref{def:typeInclusionRec-new}.\ref{leq-ax-11}
that $\record{l_i:\sigma_i \mid i\in I} + \rho_2 \leq \rho_1 + \rho_2$; then we conclude that
$\Gamma\deduces N:\sigma$ by (\ref{SR-d}) and rule $(\leq)$. \qedhere
\end{enumerate}
\end{proof}



\newcommand{\ClassCmb}{{\mathcal C}}
\newcommand{\MixinCmb}{{\mathcal M}}

\subsection{Class and Mixin combinators}\label{subsec:Mixin}

The following definition of classes and mixins is inspired by \cite{CookHC90} and \cite{BrachaC90} respectively, 
though with some departures to be discussed below. To make the description more concrete, in the examples we add constants to $\lambdaR$.

Recall that a {\em combinator} is a term in $\lambdaR^0$, namely a closed term. Let $\Y$ be  
Curry's fixed point combinator: $\lambda f.(\lambda x.f(xx))(\lambda x.f(xx))$ 
(the actual definition of $\Y$ is immaterial however, since all fixed point combinators have the same B\"ohm tree as a consequence of 
\cite{Barendregt84} Lemma 6.5.3, and hence have the same types in {\bf BCD} by the approximation theorem: see e.g.~\cite{BDS13} Theorem 12.1.17).

\begin{defi}
Let $\myClass, \State$ and $\argClass$ be (term) variables and $\Y$ be a fixed point combinator; then we define the following sets of combinators:
\[
\begin{array}{llll}
\text{Class:}  & C & ::= & \Y (\lambda \,\myClass \ \lambda \, \State. \,  \record{l_i = N_i \mid i \in I}) \\ [1mm]
\text{Mixin:} & M & ::= & \lambda \, \argClass. \, \Y (\lambda \,\myClass \, \lambda \, \State. \, (\argClass \; \State) \Override \record{l_i = N_i \mid i \in I})\\ [1mm]
\end{array}
\]
We define $\ClassCmb$ and $\MixinCmb$ as the sets of classes and mixins respectively.
\end{defi}

To illustrate this definition let us use the abbreviation $\letin{x = N} M$ for $M\Subst{N}{x}$. 
Then a  class combinator $C\in \ClassCmb$ can be written
in a more perspicuous way as follows:
\begin{equation}\label{eq:class-def}
C \equiv \Y(\lambda \,\myClass \ \lambda \, \State. \ \letin{\self = ( \myClass \; \State)}  \record{l_i = N_i \mid i \in I}).
\end{equation}

A  {\em class} is the fixed point of a function,  the {\em class definition}, mapping a recursive definition of the class itself and a state $S$, that is the value or a record of values in general, for the instance variables of the class, into a record $\record{l_i = N_i \mid i \in I}$ of {\em methods}.  A class $C$ is instantiated to an {\em object} $O \equiv C\,S$ by applying the class $C$ to
a state $S$. Hence we have:
\[
O \equiv C\,S \reduces^* \letin{\self = (C\,S)}  \record{l_i = N_i \mid i \in I},
\]
where the variable $\self$ is used in the method bodies $N_i$ to call other methods from the same object. Note that
the recursive parameter $\myClass$ might occur in the $N_i$ in subterms other than $( \myClass \; \State)$, and in particular
$N_i\Subst{C}{\myClass}$ might contain a subterm $C\,S'$, where $S'$ is a state possibly different than $S$; even $C$ itself might be returned as the value
of a method. 
Classes are the same as in \cite{CookHC90} \S 4, but for the explicit identification of $\self$ with $(\myClass \; \State)$.

We come now to typing of classes.
Let $R =  \record{l_i = N_i \mid i \in I}$, and suppose that $C \equiv  \Y (\lambda \,\myClass \ \lambda \, \State. \,  R) \in \ClassCmb$. To type $C$ we must find
a type $\sigma$ (a type of its state) and a sequence of types $\rho_1,\ldots , \rho_n \in \TTR$ such that for all $0<i< n$:
\[ 
\myClass: \sigma \to \rho_i, \State: \sigma \der R : \rho_{i+1}.
\]
Note that this is always possible for any $n$: in the worst case, we can take $\rho_i = \record{l_i:\omega \mid i\in I}$ for all $0 < i \leq n$. 
In general one has more expressive types, depending on the typings of the $N_i$ in $R$ (see example \ref{ex:class-type} below). It follows that:
\[
\der \lambda \,\myClass \ \lambda \, \State. \,  R: (\omega\to(\sigma\to\rho_1)) \inter \bigcap_{i = 1}^{n-1} ((\sigma \to \rho_i) \to (\sigma \to \rho_{i+1})),
\] 
and therefore, by using the fact that $\der \Y : (\omega\to \tau_1) \inter \cdots \inter (\tau_{n-1} \to \tau_n) \to \tau_n$ for
arbitrary types $\tau_1,\ldots,\tau_n$ (see e.g.~\cite{BDS13}, p. 586), we conclude that the typing of classes has the following shape (where $\rho =\rho_n$):
\begin{equation}\label{eq:class-typing}
\der  C \equiv \Y (\lambda \,\myClass \ \lambda \, \State. \,  \record{l_i = N_i \mid i \in I}): \sigma \to \rho
\end{equation}
In conclusion the type of a class $C$ is the arrow from the type of the state $\sigma$ to a type $\rho$ of its instances.

\begin{exa}\label{ex:class-type}

The class $\Nat$ has a method $\get$ returning the current state and a method $\succ$ returning the state of the instance object incremented by one. Beside there is a third method $\set$ whose body is the identity function:
\begin{align*}
&\Nat = \Y (\lambda \myClass. \lambda \state. \\
&\qquad \letin{\self = \myClass \ \state} \\
&\qquad \qquad \record{\get = \state, \set = \lambda \State'.\state', \succ = \self.\set(\self.\get + 1)}\\
\end{align*}

There are several ways in which functional state update can be implemented. One possibility for a method to update the current state of the object functionally is to return an updated instance of the object. As will become clear later on,
in the context of mixin applications it is unclear whether the updated object has to be an instance of the early-bound inner class ($\Nat$ in the above example) or an instance of the class modified by mixin applications. The former option loses information whenever an object that is an instance of a class modified by mixin applications is updated. The latter option breaks compatibility whenever a mixin incompatibly overwrites a method using record merge, which will be discussed later on.
In conclusion, the decision regarding the instantiation of the updated object cannot be made a-priori, but depends on the context in which the method that updates the object is called. Since in the functional setting two object instances of one class are distinguished only by the underlying state,
a method, such as $\set$, that modifies the underlying object state returns the updated state and leaves instantiation to the caller.

\medskip
To Type $\Nat$, let us call 
\[ R \equiv \record{\get = \state, \set = \lambda \State'.\state', \succ = \self.\set(\self.\get + 1)} \]
where $\self \equiv (\myClass \ \state)$. Now we can deduce:
\[ 
\myClass : \omega, \state : \Int \vdash R : \record{\get:\Int, \set:\Int\to\Int, \succ:\omega} = \rho_1.
\]
since $\myClass : \Int\to\omega, \state : \Int \vdash \myClass \ \state : \omega$ and therefore $\self.\set(\self.\get + 1)$ is just typable by
$\omega$; on the other hand the type of $\set$ is any type of the identity, so that $\Int\to\Int$ is a possibility. 

Given $\rho_1$ we consider the new assumption $\myClass : \Int \to \rho_1$ by which we obtain the typing $\self \equiv \myClass \ \state : \rho_1$
that is enough to get $\self.\set(\self.\get + 1): \Int$ and hence:
\[
\myClass : \Int \to \rho_1, \state : \Int \vdash R : \record{\get:\Int, \set:\Int\to\Int, \succ:\Int} = \rho_2.
\]
Eventually by (\ref{eq:class-typing}) we conclude that $\Nat: \Int \to \rho_2$.

\medskip
In case of $\Nat$ we have reached the type $\rho_2$ that doesn't contain any occurrence of $\omega$ in a finite number of steps. 
This is because method $\succ$ of class $\Nat$ returns a number; let us consider a variant $\Nat'$ of $\Nat$ having a method
$\succ$ that returns a new instance of the class, whose state has been incremented:

\begin{align*}
&\Nat' = \Y (\lambda \myClass. \lambda \state. \\
&\qquad \letin{\self = \myClass \ \state} \\
&\qquad \qquad \record{\get = \state,  \succ =  \myClass(\self.\get +1)}\\
\end{align*}
Let us call $R' \equiv \record{\get = \state,  \succ =  \myClass(\self.\get +1)}$, where $\self \equiv (\myClass \ \state)$ as above.
Then
\[
\myClass : \Int\to\omega, \state : \Int \vdash R' : \record{\get:\Int, \succ:\omega} = \rho'_1.
\]
By assuming $\myClass : \Int \to \rho'_1$ and $\state : \Int$, we have that $\self: \rho'_1$ and hence $\self.\get:\Int$. It follows that
$\myClass(\self.\get +1) : \rho'_1$ so that
\[
\myClass : \Int \to \rho'_1, \state : \Int \vdash R' : \record{\get:\Int, \succ:\rho'_1} = \rho'_2.
\]
In general, putting $\rho'_0 = \omega$ and $\rho'_{i+1} = \record{\get:\Int, \succ:\rho'_i}$ we have for all $i$:
\[
\myClass : \Int \to \rho'_i, \state : \Int \vdash R' : \record{\get:\Int, \succ:\rho'_i} = \rho'_{i+1}.
\]
From this by (\ref{eq:class-typing}) we conclude that $\Nat': \Int \to \rho'_i$ for all $i$.  We note that this time, differently than in case of $\Nat$, we cannot get rid of occurrences $\omega$ in the $\rho'_i$.

A variant of this typing uses the type constants $\Odd$ and $\Even$ called semantic types in \cite{JR13} where they are used in the intersection types $\Int \inter \Odd$ or $\Int \inter \Even$
(see also below Section \ref{subsec:semanticTypes}). Let us suppose
that we have axioms $x: \Int \inter\Odd \vdash x + 1: \Int \inter\Even$ and $x: \Int \inter\Even \vdash x + 1:\Int \inter\Odd$ added to the typing system. Then a more interesting type for $\Nat'$ is
\begin{align*}
\Nat' : \quad&(\Int \inter\Odd \to \record{\get:\Int \inter\Odd, \succ:\record{\get:\Int \inter\Even}}) ~ \inter \\ 
	 &(\Int \inter\Even \to \record{\get:\Int \inter\Even, \succ:\record{\get:\Int \inter\Odd}}) .
\end{align*}

Although $\Nat'$ is closer to what is done with object-oriented programming, in the context of mixin application, introduced in the remainder of this section, we shall not consider this kind of recursive definition in this work, as previously explained.

\end{exa}

\medskip A {\em mixin} $M \in \MixinCmb$ is a combinator such that, if $C\in \ClassCmb$ then $M\,C$ reduces to a new class $C' \in \ClassCmb$.
Writing $M$ in a more explicit way we obtain:
\begin{align*}\label{eq:mixin-def}
M \equiv  \lambda \, \argClass. \, \Y ( \lambda \,\myClass \, \lambda \, \State. \,&\letin{ \super = (\argClass\;\State)} \\
	&\letin{\self = ( \myClass \; \State)}\\ 
	&\super \Override \record{l_i = N_i \mid i \in I})
\end{align*}
In words, a mixin merges an instance $C\,S$ of the input class $C$ with a new state $S$ together with a {\em difference} record 
$R \equiv \record{l_i = N_i \mid i \in I}$, that would be written $\Delta(C\,S)$ in terms of \cite{BrachaC90}. Note that our mixins are not the same as class modificators (also called wrappers e.g.~in \cite{BrachaThesis}). Wrappers bind $\super$ to the unmodified class definition applied to $\self$ \emph{without} taking the fixed point. In our case, $\super$ is simply an instance of the unmodified class.
The effect is that we have a static (or early) binding instead of dynamic (or late) binding of $\self$.

Let $M \equiv  \lambda \, \argClass. \, \Y (\lambda \,\myClass \, \lambda \, \State. \, (\argClass \; \State) \Override R) \in \MixinCmb$; 
to type $M$ we have to find types $\sigma^1, \sigma^2, \rho^1$ and a sequence $\rho^2_1, \ldots, \rho^2_n \in \TTR$ of record types such that for all $1 \leq i < n$ it is true that
$\lbl(R) = \lbl(\rho^2_{i})$ and setting
\begin{align*}
\Context_0 &= \Set{\argClass: \sigma^1 \to \rho^1, \myClass: \omega, \State: \sigma^1\inter\sigma^2}\\
\Context_{i} &= \Set{\argClass: \sigma^1 \to \rho^1, \myClass: (\sigma^1\inter\sigma^2) \to \rho^1 +  \rho_{i}^2, \State: \sigma^1\inter\sigma^2} \text{ for all } 1 \leq i < n
\end{align*}
we may deduce for all $0 \leq i < n$:
\[
\prooftree 
	\prooftree
		\prooftree
			\Context_i \der \State: \sigma^1\inter\sigma^2
		\justifies
			\Context_i \der \State: \sigma^1
		\using (\leq)
		\endprooftree
	\justifies
		\Context_i \der \argClass \; \State:\rho^1
		\using (\ArrE)
	\endprooftree
	\qquad
	\Context_i  \der R: \rho_{i+1}^2 \qquad \lbl(R) = \lbl(\rho^2_{i+1})
\justifies
	\Context_i \der (\argClass \; \State) \Override R: \rho^1+\rho_{i+1}^2
\using (+)
\endprooftree
\]
Hence for all $0\leq i<n$ we can derive the typing judgment:
\begin{eqnarray*}
\lefteqn{\argClass: \sigma^1 \to \rho^1 \der \lambda \,\myClass \, \lambda \, \State. \, (\argClass \; \State) \Override R:} \\ [2mm]
& & \hspace{4.5cm}((\sigma^1\inter\sigma^2) \to (\rho^1 +  \rho_i^2)) \to (\sigma^1\inter\sigma^2) \to ( \rho^1+\rho_{i+1}^2)
\end{eqnarray*}
and therefore, by reasoning as for classes, we get (setting $\rho^2 = \rho_{n}^2$):
\begin{multline}\label{eq:mixin-typing}
\der  M \equiv \lambda \, \argClass. \, \Y (\lambda \,\myClass \, \lambda \, \State. \, (\argClass \; \State) \Override R): \\
(\sigma^1 \to \rho^1) \to (\sigma^1\inter\sigma^2) \to ( \rho^1+\rho^2)
\end{multline}
Spelling out this type, we can say that $\sigma^1$ is a type of the state of the argument-class of $M$; $\sigma^1\inter\sigma^2$ is the type
of the state of the resulting class, that refines $\sigma^1$. Type $\rho^1$ expresses the requirements of $M$ about the methods of the argument-class, i.e.~what is assumed to hold for the usages of $\super$ and $\argClass$ in $R$ to be properly typed; $\rho^1+\rho^2$ is a type of the record of methods of the refined class, resulting from the merge of the methods of the argument-class with those of the difference $R$; since in general there will be overridden methods, whose types might be incompatible, the $+$ type constructor cannot be replaced by intersection. 
On the other hand, $M$ preserves typings of any label that is not in $R$.
\begin{lem}
\label{lem:mixin-label-preservation}
For any $\sigma$ and $\rho = \record{l : \tau}$ such that $l \not\in \lbl(R)$ we have 
$$\der M \equiv \lambda \, \argClass. \, \Y (\lambda \,\myClass \, \lambda \, \State. \, (\argClass \; \State) \Override R): (\sigma \to \rho) \to (\sigma \to \rho)$$
\end{lem}
\begin{proof}
Let $\{l_1, \ldots, l_n\} = \lbl(R)$. We have $\der R : \record{l_1 : \omega, \ldots, l_n : \omega}$. Since $l \not\in \lbl(R)$ we obtain $\rho + \record{l_1 : \omega, \ldots, l_n : \omega} =  \record{l : \tau, l_1 : \omega, \ldots, l_n : \omega} \leq \rho$. Using the above general type derivation with $\sigma^1 = \sigma^2 = \sigma$, $\rho^1 = \rho$, $\rho^2 = \record{l_1 : \omega, \ldots, l_n : \omega}$ and $\der \Y : (\omega \to (\sigma \to \rho)) \to (\sigma \to \rho)$ we obtain $\der M : (\sigma \to \rho) \to (\sigma \to \rho + \rho^2) \leq (\sigma \to \rho) \to (\sigma \to \rho)$.
\end{proof}

\begin{exa}\label{ex:mixin-type}
Let us consider the mixin $\Comparable$ adding a method $\compare$ to its argument class, which is supposed to have a method $\get$:
\begin{align*}
&\Comparable = \lambda \argClass. \Y (\lambda \myClass. \lambda \state. \\
&\qquad \letin{\super = \argClass \ \state} \\
&\qquad \letin{\self = \myClass \ \state} \\
&\qquad \qquad \super \oplus \record{\compare = \lambda o.(o.\get == \self.\get)})
\end{align*}
We write $==$ for a suitable equality operator, which is distinct from the symbol $=$ used in our calculus to associate labels to their values in a record.

For the sake of readability, we take the types $\sigma^1$ and $\sigma^2$ of the respective states of the argument class and the resulting class to be just $\Int$ and simply write $\Int$ for $\Int\inter\Int$ which is the type we assume for $\state$. If we set
\[\Context_0 = \Set{\argClass:\Int\to\record{\get:\Int}, \myClass:\Int\to\omega, \state:\Int }\] 
we have $\Context_0\vdash \super:\record{\get:\Int}$ but $\Context_0\vdash \self:\omega$ so that
$\Context_0\vdash \lambda o.(o.\get == \self.\get):\tau \to \omega$ for any type $\tau$ given to $o$; 
but $\tau \to \omega = \omega$ by Definition
\ref{def:typeInclusionInt}, so that we conclude
\[
\Context_0\vdash \super \oplus \record{\compare = \lambda o.(o.\get == \self.\get)} : \record{\get:\Int} + \record{\compare:\omega}.
\]
Now, taking
\[
\Context_1 = \Set{\argClass:\Int\to\record{\get:\Int}, \myClass:\Int\to \record{\get:\Int} + \record{\compare:\omega}, \state:\Int }
\]
and using the fact that $ \record{\get:\Int} + \record{\compare:\omega} =  \record{\get:\Int} \inter \record{\compare:\omega} \leq  \record{\get:\Int}$
we have $\Context_1\vdash \self:\record{\get:\Int}$, that implies  $\Context_1\vdash \self.\get:\Int$
and therefore $\Context_1 \vdash \lambda o.(o.\get == \self.\get): \record{\get:\Int} \to \Bool$. By (\ref{eq:mixin-typing}) we conclude that
$\Comparable$ has type
\[
(\Int\to\record{\get:\Int}) \to \Int \to \record{\get:\Int} + \record{\compare:\record{\get:\Int} \to \Bool}.
\]

We finally observe that if the argument class $C$ has type $\Int\to\record{\get:\Int, \compare:\tau}$ for some type $\tau$, then we have
that $\Int\to\record{\get:\Int, \compare:\tau} \leq \Int\to\record{\get:\Int}$, so that the class $C' \equiv \Comparable\, C$ has the correct type
$\Int \to \record{\get:\Int} + \record{\compare:\record{\get:\Int} \to \Bool}$ as the method $\super.\compare$ is overridden in 
$\super \oplus \record{\compare = \lambda o.(o.\get == \self.\get)}$. This is 
because, even by keeping the typing $\super: \record{\get:\Int, \compare:\tau}$ through the derivation, we obtain that $\super \oplus \record{\compare = \lambda o.(o.\get == \self.\get)}$ has type
\[\record{\get:\Int, \compare:\tau} + \record{\compare:\record{\get:\Int} \to \Bool} = \record{\get:\Int, \compare:\record{\get:\Int} \to \Bool}\]
which is the same as $\record{\get:\Int} + \record{\compare:\record{\get:\Int} \to \Bool}$.
For a more general typing of $\Comparable$ see Appendix \ref{app:running}.

\end{exa}


	\section{Bounded Combinatory Logic}\label{sec:bcl}
	
	Our main goal is synthesize meaningful mixin compositions.
For this purpose, we use the logical programming language given by inhabitation in $\bclkc$ (bounded combinatory logic with constructors). Conveniently, necessary type information can be inferred from intersection types for $\lambdaR$.
 $\bclkc$ is an extension of bounded combinatory logic $\bclk$~\cite{RehofEtAlCSL12} by covariant constructors. Constructors can be utilized to encode record types and useful features of $+$ while extending the existing synthesis framework (CL)S~\cite{BDDMR14}. In this section, we present $\bclkc$ along with its fundamental properties, in particular decidability of inhabitation.

\subsection{\texorpdfstring{$\bclkc$}{BCLk(TC)}}

\emph{Combinatory terms} are formed by application of combinators from a \emph{repository} (combinatory logic context) $\Delta$.
\begin{defi}[Combinatory Term]
$E, E' ::= C\; |\; (E\; E') \quad \text{ where } C \in \dom(\Delta)$
\end{defi}
We create repositories of typed combinators that can be considered logic programs for the existing synthesis framework (CL)S~\cite{BDDMR14} to reason about semantics of such compositions. The underlying type system of (CL)S is the intersection type system {\bf BCD} \cite{bcd}, which we extend by covariant constructors. The extended type system $\TTC$, while suited for synthesis, is flexible enough to encode record types and features of $+$. 

\begin{defi}[Intersection Types with Constructors $\TTC$]
The set $\TTC$ is given by:
\[
\TTC \ni \sigma, \tau, \tau_1, \tau_2 \ ::= a \mid \alpha \mid \omega \mid \tau_1 \rightarrow \tau_2 \mid \tau_1 \cap \tau_2 \mid c(\tau)
\]
where $a$ ranges over constants, $\alpha$ over type variables and $c$ over unary constructors $\CC$.
\end{defi}

$\TTC$ adds the following two subtyping axioms to the {\bf BCD} system (cf. Definition \ref{def:typeInclusionInt})
$$\tau_1 \leq \tau_2 \Rightarrow
c(\tau_1) \leq c(\tau_2) \qquad \quad
c(\tau_1) \cap c(\tau_2) \leq 
c(\tau_1 \cap \tau_2)$$
The additional axioms ensure \emph{constructor distributivity}, i.e., $c(\tau_1) \cap c(\tau_2) = c(\tau_1 \cap \tau_2)$. 

The notion of \emph{level} of a type is used as a bound to ensure decidability of inhabitation, or equivalently termination of logic programs in (CL)S.

\begin{defi}[Level] We define the \emph{level} of a type as follows:
\begin{align*}
&\level(\omega) = \level(a) = \level(\alpha) = 0
&&\level(c(\tau)) = 1 + \level(\tau) \\
&\level(\sigma \to \tau) = 1 + \max(\level(\sigma), \level(\tau))
&&\level(\sigma \cap \tau) = \max(\level(\sigma), \level(\tau))
\end{align*}
We define the level of a substitution $S$ as $\level(S)=\max\{\level(S(\alpha)) \mid \alpha \in \dom(S)\}$.
\end{defi}

\begin{defi}[Type Assignment $\bclkc$]
\[
\begin{array}{ccc }
\prooftree
	C:\tau \in \Delta \qquad \level(S) \leq k
	\justifies
	\Delta \vdashk C:S(\tau)
	\using (\text{Var})
\endprooftree 
&\quad
\prooftree
	\Delta \vdashk E: \sigma\to\tau \qquad \Delta \vdashk E':\sigma
	\justifies
	\Delta \vdashk E E':\tau
	\using (\ArrE)
\endprooftree 

\\\\

\prooftree
	\Delta \vdashk E: \sigma \quad 
	\Delta \vdashk E: \tau
	\justifies
	\Delta \vdashk E: \sigma\inter\tau
	\using (\inter)
\endprooftree 
&\quad
\prooftree
	\Delta \vdashk E: \sigma \qquad \sigma\leq\tau
	\justifies
	\Delta \vdashk E : \tau
	\using (\leq)
\endprooftree 
\end{array}\]
\end{defi}

For a set of typed combinators $\Delta$ and a type $\tau \in \TTC$ we say $\tau$ is inhabited in $\Delta$, if there exists a combinatory term $E$ and a $k \in \NN$ such that $\Delta \vdashk E : \tau$.

\subsection{\texorpdfstring{$\bclkc$}{BCLk(TC)} Inhabitation}

In this section we extend the understanding of $\bcl$ inhabitation~\cite{RehofEtAlCSL12} to constructors. 
First, we extend the necessary notions of \emph{paths} and \emph{organized types} to constructors in the following way.

\begin{defi}[Path]
A \emph{path} $\pi$ is a type of the form:
$ \pi ::= a \mid \alpha \mid \sigma \to \pi \mid c(\omega) \mid c(\pi) $,
where $\alpha$ is a variable, $\tau$ is a type, $c$ is a constructor and $a$ is a constant.
\end{defi}

\begin{defi}[Paths in $\tau$] 
\label{def:ttc_paths_in_tau}
Given a type $\tau \in \TTC$, the set $\PP(\tau)$ of paths in $\tau$ is defined as
\begin{align*}
&\PP(a) = \{a\} && \PP(\sigma \to \tau) = \{\sigma \to \pi \mid \pi \in \PP(\tau)\}\\
&\PP(\alpha) = \{\alpha\} && \PP(\sigma \cap \tau) = \PP(\sigma)\cup \PP(\tau)\\
& \PP(\omega) = \emptyset && \PP(c(\tau)) = \begin{cases} \{c(\omega)\} & \text{if } \PP(\tau) = \emptyset \\ \{c(\pi) \mid \pi \in \PP(\tau)\} & \text{else} \end{cases}
\end{align*}
\end{defi}
\noindent Observe that $\PP(\tau) = \emptyset$ iff $\tau=\omega$.

\begin{defi}[Organized Type]
A type $\tau$ is called organized, if $\tau \equiv \bigcap_{i \in I} \pi_i$, where $\pi_i$ for $i \in I$ are paths.
\end{defi}

\noindent For brevity, we sometimes write $\bigcap \PP(\tau)$ for $\pi_1\cap\ldots\cap\pi_n$ where $\PP(\tau)=\{\pi_1,\ldots,\pi_n \}$. If $\PP(\tau)=\emptyset$, then we set $\bigcap \PP(\tau)\equiv\omega$.

\begin{lem}
\label{lem:ttc_path_intersection}
Given a type $\tau \in \TTC$, the type $\bigcap \PP(\tau)$ is organized with $\bigcap \PP(\tau) = \tau$.
\end{lem}
\noindent The detailed proof by induction of Lemma \ref{lem:ttc_path_intersection} is in the appendix. 

\begin{lem}
\label{lem:ttc_path_intersection_size}
Given a type $\tau \in \TTC$ we have $|\bigcap \PP(\tau)| \leq |\tau|^2$, where $|\cdot|$ denotes the number of nodes in the syntax tree of a given type.
\end{lem}

\begin{ProofNoQed}
By induction using Definition \ref{def:ttc_paths_in_tau} we have $|\PP(\tau)| \leq |\tau|$.
The only non-trivial cases for the inductive proof of the main statement are (assuming $\PP(\tau) \neq \emptyset$)
\begin{align*}
|\bigcap \PP(c(\tau))| &= \sum\limits_{\pi \in \PP(\tau)} (|\pi|+2) - 1 = |\PP(\tau)| + |\bigcap \PP(\tau)| \leq |\tau| + |\tau|^2 \leq |c(\tau)|^2\\
|\bigcap \PP(\sigma \to \tau)| &= \sum\limits_{\pi \in \PP(\tau)} (|\sigma| + |\pi|+2) - 1 = (|\sigma|+1) \cdot |\PP(\tau)| + |\bigcap \PP(\tau)| \leq |\sigma \to \tau|^2 \tag*{\qEd}
\end{align*}

\end{ProofNoQed}

Due to Lemma \ref{lem:ttc_path_intersection_size}, for any intersection type there exists an equivalent organized intersection type computable in polynomial time. 
Note that organized types are not necessarily normalized~\cite{Hindley82} or strict~\cite{Bakel11}. However, organized types have the following property known from normalized types. 

\begin{lem}\label{lem:ttc_path_set_subtyping}
Given two types $\sigma, \tau \in \TTC$, we have $\sigma \leq \tau$ iff for each path $\pi \in \PP(\tau)$ there exists a path $\pi' \in \PP(\sigma)$ such that $\pi' \leq \pi$ and
\begin{itemize}
\item If $\pi \equiv \alpha$ (resp. $a$), then $\pi' \equiv \alpha$ (resp. $a$).
\item If $\pi \equiv \sigma_2 \to \tau_2$, then $\pi' \equiv \sigma_1 \to \tau_1$ such that $\sigma_2 \leq \sigma_1$ and $\tau_1 \leq \tau_2$.
\item If $\pi \equiv c(\tau_1)$, then $\pi' \equiv c(\sigma_1)$ such that $\sigma_1 \leq \tau_1$.
\end{itemize}
\end{lem}
\noindent The detailed proof by induction of Lemma \ref{lem:ttc_path_set_subtyping} is in the appendix.

We say a path $\pi$ has the \emph{arity} of at least $m$ if $\pi \equiv \sigma_1 \to \ldots \to \sigma_m \to \tau$ for some $\sigma_1, \ldots \sigma_m, \tau \in \TTC$.
Additionally, for such a path $\pi$ we define $\arg_i(\pi)=\sigma_i$ for $1 \leq i \leq m$ and $\tgt_m(\pi)=\tau$.
We define $\PP_m(\tau)$ as the set of paths in $\tau$ having arities of at least $m$.

$\bclk$ inhabitation is $(k+2)$-\textsc{ExpTime} complete~\cite{RehofEtAlCSL12}. The upper bound is derived by constructing an alternating Turing machine based on the \emph{path lemma}~\cite[Lemma~11]{RehofEtAlCSL12}. Accordingly, we formulate a path lemma for $\bclkc$. Let $\atoms(\tau)$ be the set of constants, variables and constructor names occurring in $\tau$. Additionally, for a substitution $S$ let $\atoms(S)=\bigcup \{\atoms(S(\alpha)) \mid \alpha \in \dom(S)\}$, and for a repository $\Delta$ let $\atoms(\Delta) = \bigcup \{\atoms(\tau) \mid C : \tau \in \Delta\}$.

\begin{lem}[Path Lemma for $\bclkc$] 
\label{lem:blckc_path_lemma}
The following are equivalent conditions:
\begin{enumerate}
\item $\Delta \vdashk C\, E_1\, \ldots E_m : \tau$
\item There exists a set of paths \\
$P \subseteq \PP_m(\bigcap \{S(\Delta(C)) \mid \level(S) \leq k, \atoms(S) \subseteq \atoms(\Delta) \cup \atoms(\tau)\})$ such that
\begin{enumerate}
\item $\bigcap_{\pi \in P} \tgt_m(\pi) \leq \tau$
\item $\Delta \vdashk E_i : \bigcap_{\pi \in P} \arg_i(\pi)$ for $1 \leq i \leq m$
\end{enumerate}
\end{enumerate}
\end{lem}
\noindent The detailed proof of Lemma \ref{lem:blckc_path_lemma}, which is a slight extension of~\cite[Lemma 11]{RehofEtAlCSL12}, is in the appendix.

In the above Lemma \ref{lem:blckc_path_lemma} we can bound the size of the set $P$ of paths of level at most $k$. Let $\exp_k$ be the iterated exponential function, i.e.~$\exp_0(n) = n$ and $\exp_{k+1}(n) = 2^{\exp_k(n)}$.

\begin{lem}
\label{lem:polynomial_size_bound}
There exists a polynomial $p$ such that (modulo $=$) for any $k \in \NN$ the number of level-$k$ types
over $n$ atoms is at most $\exp_{k+1}(p(n))$, and the size of such types is at most $\exp_k(p(n))$.
\end{lem}

\begin{proof}
Since normalizing (i.e.~recursively organizing) a type does not increase its level, we only need to consider normalized types. Fix the set of atoms $\AA$ with $|\AA|=n$. Let $\PP^k$ denote the set of paths of level at most $k$, and $\TT^k$ the set of normalized types of level at most $k$. We have
\begin{align*}
& |\AA| = n \leq n|\TT^k|^2 \text{ and }\\ 
 & |\{\sigma \to \pi \mid \sigma \in \TT^k, \pi \in \PP^k\}| \leq |\TT^k| \cdot |\PP^k| \leq n|\TT^k|^2 \text{ and }\\
 & |\{c(\pi) \mid c \in \AA, \pi \in \PP^k\}| \leq |\AA| \cdot |\PP^k| \text{ therefore }\\
 & |\PP^{k+1}| \leq |\AA| + |\{\sigma \to \pi \mid \sigma \in \TT^k, \pi \in \PP^k\}| + |\{c(\pi) \mid c \in \AA, \pi \in \PP^k\}| \leq 3n|\TT^k|^2 \text{ and }\\
 & |\TT^{k+1}| \leq |\{\bigcap P \mid P \subseteq \PP^{k+1}\}| \leq 2^{3n|\TT^k|^2}
\end{align*}
By induction on $k$ we have that there exists a polynomial $p$ such that 
$$|\TT^k| \leq \exp_{k+1}(p(n)+\sum\limits_{i=0}^{k-1}\frac{p(n)}{2^i}) \leq \exp_{k+1}(3p(n))$$
Let $s_k$ be the maximal size of a type of level $k$ with atoms in $\AA$. We have 
$$s_{k+1} \leq |\PP^{k+1}|\cdot(2s_{k}+2) \leq 3n|\TT^{k}|^2\cdot(2s_{k}+2) $$
By induction on $k$ we have that there exists a polynomial $p$ such that $s_k \leq \exp_k(p(n))$.
\end{proof}
For the interested reader, only the rank of a given type influences the height of the exponentiation tower.

Using Lemma \ref{lem:blckc_path_lemma}  we can decide  $\bclkc$ inhabitation by the alternating Turing machine shown in Figure \ref{fig:atm_bclkc}.

\begin{figure}[!ht]
\caption{Alternating Turing machine deciding inhabitation in $\bclkc$}
\label{fig:atm_bclkc}
\[
\begin{array}{ll} 
 & \mathit{Input}: \ \  \Delta, \tau, k \\ 
  & \mathit{loop:} \\
1 & \mbox{\sc choose} \ (C:\sigma) \in \Delta \\
2 & \sigma' := \bigcap \{S(\sigma) \mid\level(S) \leq k, \atoms(S) \subseteq \atoms(\Delta) \cup \atoms(\tau)\} \\
3 & \mbox{\sc choose} \ m \in \{0,\ldots, \text{maximal arity of paths in } \sigma' \} \\
4 & \mbox{\sc choose} \ P \subseteq \mathbb{P}_m(\sigma') \\
  & \\
5 & \mbox{\sc if} \ (\bigcap_{\pi \in P} \mathit{tgt}_m(\pi) \le \tau ) \ \mbox{\sc then} \\
6 & \ \ \mbox{\sc if } (m=0) \mbox{ {\sc then accept}}  \\
7 & \ \ \mbox{{\sc else} } \\
8 & \ \ \ \ \mbox{\sc forall}(i = 1 \ldots m) \\
9 &  \ \ \ \ \ \ \tau := \bigcap_{\pi \in P} \mathit{arg}_i(\pi) \\
10 & \ \ \ \ \ \ \mbox{\sc goto } \mathit{loop} \\
\end{array}
\]
\end{figure}

\begin{thm}
$\bclkc$ inhabitation in $(k+2)$-\textsc{ExpTime}.
\end{thm}

\begin{proof}
The alternating Turing machine in Figure \ref{fig:atm_bclkc} directly implements Lemma \ref{lem:blckc_path_lemma} and is therefore sound and complete. To show that the machine operates in alternating $(k + 1)$-\textsc{ExpSpace}, we need to bound the size of (organized) $\sigma'$, which is $(n \cdot \exp_{k+1}(p(n)) \cdot \exp_{k}(p(n)))^2$ due to Lemmas \ref{lem:polynomial_size_bound} and \ref{lem:ttc_path_intersection_size}. By the identity \textsc{ASpace}$(f(n))$ $=$ \textsc{DTime}$(2^{\calO(f(n))})$ \cite{chandra1981alternation}, $\bclkc$ inhabitation is in $(k + 2)$-\textsc{ExpTime}.
\end{proof}


	\section{Mixin Composition Synthesis by Type Inhabitation}\label{sec:synthesis}
	
	In this section, we present an encoding of record types by $\TTC$ types that capture mixin semantics. We use the encoding to define a repository of typed combinators from classes and mixins that can be used to synthesize meaningful mixin compositions by means of $\bclkc$ inhabitation. In the following we fix a finite set of labels $\mathcal{L} \subseteq \Label$ that are used in the particular domain of interest for mixin composition synthesis.

\subsection{Records as Unary Covariant Distributing Constructors}
First, we need to encode record types as $\TTC$ types.
We define constructors $\rrecord{\cdot}$ and $l(\cdot)$ for $l \in \mathcal{L}$ to represent record types using the following partial translation function $\interpret{\cdot} \colon \TT \to \TTC$ as follows:
\begin{align*}
\interpret{\omega} & = \omega
& \interpret{a} & = a \\
\interpret{\sigma \to \tau} & = \interpret{\sigma} \to \interpret{\tau} 
&\interpret{\sigma \cap \tau} & = \interpret{\sigma} \cap \interpret{\tau} \\
\interpret{\record{l : \tau}} & = \rrecord{l(\interpret{\tau})}
& \interpret{\record{}} & = \rrecord{\omega}
\end{align*}
By definition, we have $\interpret{\record{l_i : \tau_i \mid i \in I}} = \interpret{\bigcap_{i \in I} \record{l_i : \tau_i}} = \bigcap_{i \in I} \rrecord{l_i(\interpret{\tau_i})}$ if $I \neq \emptyset$.
\begin{lem}
\label{lem:interpret-bijection}
For any $\sigma, \tau \in \TT$ such that $\interpret{\sigma}$ and $\interpret{\tau}$ are defined we have $\sigma = \tau$ iff $\interpret{\sigma} = \interpret{\tau}$.
\end{lem}
\begin{proof}
Routine induction on $\leq$ derivation observing that $\sigma, \tau$ do not contain $+$, $\interpret{\cdot}$ is homomorphic wrt. $\to$ and $\cap$, and atomic records are covariant and distribute over $\cap$.
\end{proof}

The translation function $\interpret{\cdot}$ is not defined for types containing $+$. Since $+$ has non-monotonic properties, it cannot be immediately represented by a covariant type constructor. Simply applying Lemma \ref{lem:normal-recordTypes}.\ref{lem:normal-recordTypes-c} would require to consider all arguments a mixin could possibly be applied to. Such an unwieldy specification would not be adequate for synthesis.
There are two possibilities to deal with this problem. 
The first option is extending the type-system used for inhabitation. Here, the main difficulty is that existing inhabitation algorithms rely on the separation of intersections into paths \cite{RehofEtAlCSL12}. As demonstrated in the remark accompanying Lemma \ref{lem:typePlus}, it becomes unclear how to perform such a separation in the presence of the non-monotonic $+$ operation.
The second option, pursued in the rest of this section, is to use the expressiveness of 
schematism provided by $\bclkc$. Specifically, encoding particular $\TT$ types containing $+$ that capture mixin semantics as $\TTC$ types suited for $\bclkc$ inhabitation. Ultimately, we are able to achieve completeness (cf. Theorem \ref{thm:partial_completeness}) wrt. particular mixin typings (cf. Property $(\star)$) which restrict the general shape (\ref{eq:mixin-typing}). This restriction allows for a concise schematic mixin specification suited for synthesis by $\bclkc$ inhabitation.

Let $M \equiv  \lambda \, \argClass. \, \Y (\lambda \,\myClass \, \lambda \, \State. \, (\argClass \; \State) \Override R) \in \MixinCmb$ be such that for some $\sigma, \rho_1, \rho_2$, which do not contain $+$, with $\lbl(\rho_2)=\lbl(R)$ and all $\rho \in \TTR$ we have 
\begin{align*}
\vdash M : (\sigma \to \rho \cap \rho_1) \to (\sigma \to \rho+\rho_2) \tag{$\star$}
\end{align*}
We define the $\TTC$ type 
$$\tau_M \equiv \big((\interpret{\sigma} \to \interpret{\rho_1}) \to (\interpret{\sigma} \to \interpret{\rho_2})\big) \cap \bigcap\limits_{l \in \calL \setminus \lbl(R)} \big((\interpret{\sigma} \to \rrecord{l(\alpha_l)}) \to (\interpret{\sigma} \to \rrecord{l(\alpha_l)})\big)$$
Note that $\tau_M$ contains only labels $l\in\calL$ and type variables $\alpha_l$ where $l\in\calL$ because $\sigma,\rho_1,\rho_2$ do not contain type variables (cf. Definition~\ref{def:intertype:lambdaR}).

\begin{lem}[Translation Soundness]
\label{lem:enc_soundness}
Assume $(\star)$.\\
For any substitution $S$ and type $\tau \in \TT$ such that $\interpret{\tau} = S(\tau_M)$ we have $\vdash M : \tau$.
\end{lem}

\begin{proof}
We necessarily have \\
$S(\tau_M) = \interpret{\big((\sigma \to \rho_1) \to (\sigma \to \rho_2)\big) \cap \bigcap\limits_{l \in \calL \setminus \lbl(R)} \big((\sigma \to \record{l : \tau_l}) \to (\sigma \to \record{l : \tau_l})\big)}$ for some $\tau_l$ for $l \in \calL \setminus \lbl(R)$. Due to $(\star)$ with $\rho = \record{}$ we have $\vdash M : (\sigma \to \rho_1) \to (\sigma \to \rho_2)$, and by Lemma \ref{lem:mixin-label-preservation} we have $\vdash M : (\sigma \to \record{l : \tau_l}) \to (\sigma \to \record{l : \tau_l})$ for any $l \in \calL \setminus \lbl(R)$, thus showing the claim.
\end{proof}

Translation Soundness forces instantiations of $\tau_M$ to remain within typings of $M$ in $\lambdaR$. \emph{Negative information}, i.e.~information about labels absent in $R$, is encoded by explicitly capturing all positive information, excluding labels in $\lbl(R)$, in instances of $$\bigcap\limits_{l \in \calL \setminus \lbl(R)} \big((\interpret{\sigma} \to \rrecord{l(\alpha_l)}) \to (\interpret{\sigma} \to \rrecord{l(\alpha_l)})\big)$$
Schematism, i.e.~the possibility to instantiate type variables $\alpha_l$, is essential to capture mixin behavior.
This encoding is possible under the assumption of a finite set of labels $\calL$, which is valid because synthesis does not introduce new labels. The encoding overhead is polynomially bounded by product of the number of mixin combinators and the number of labels.

\begin{lem}[Translation Completeness]
\label{lem:enc_completeness}
Assume $(\star)$.\\
For any $\rho \in \TTR$ with $\lbl(\rho) \subseteq \calL$ there exists a substitution $S$ and a type $\tau \in \TT$ such that $S(\tau_M) \leq \interpret{\tau}$ and $\tau = (\sigma \to \rho \cap \rho_1) \to (\sigma \to \rho+\rho_2)$.
\end{lem}

\begin{ProofNoQed}
Let $\rho = \bigcap\limits_{l \in L} \record{l : \tau_l}$ for some $L \subseteq \calL$. By Lemma \ref{lem:normal-recordTypes}.\ref{lem:normal-recordTypes-c} we may assume for each $l \in L$ that $\tau_l$ does not contain $+$. We define 
$$S(\alpha_l) = \begin{cases}
\interpret{\tau_l} & \text{if } l \in L \setminus \lbl(R)\\
\omega & \text{else}
\end{cases}$$
\noindent and successively applying Lemma \ref{lem:interpret-bijection} obtain
\begin{align*}
S(\tau_M) & \leq \interpret{\big((\sigma \to \rho_1) \to (\sigma \to \rho_2)\big) \cap \bigcap\limits_{l \in L \setminus \lbl(R)} \big((\sigma \to \record{l : \tau_l}) \to (\sigma \to \record{l : \tau_l})\big)}\\
& \leq \interpret{\big((\sigma \to \bigcap\limits_{l \in L \setminus \lbl(R)} \record{l : \tau_l} \cap \rho_1) \to (\sigma \to \bigcap\limits_{l \in L \setminus \lbl(R)} \record{l : \tau_l} \cap \rho_2)\big)} \\
& \leq \interpret{\big((\sigma \to \rho \cap \rho_1) \to (\sigma \to \underbrace{\bigcap\limits_{l \in L \setminus \lbl(R)} \record{l : \tau_l} \cap \rho_2}_{= \rho + \rho_2 \text{ by Lem. } \ref{lem:label}.\ref{lem:label-b}})\big)}
\tag*{\qEd}
\end{align*}
\end{ProofNoQed}

Translation Completeness ensures that each type $\lambdaR$ of $M$ according to $(\star)$ is captured by some instance of $\tau_M$.

Using the above translation properties, we construct a repository of typed combinators representing classes and mixins.

\subsection{Mixin Composition}
In this section we denote type assignment in $\lambdaR$ by $\vdashr$ and fix the following ingredients:
\begin{itemize}
\item A finite set of classes $\mathcal{C}$.
\item For each $C \in \mathcal{C}$ types $\sigma_C \in \TT, \rho_C \in \TTR$ such that $\interpret{\sigma_C \to \rho_C}$ is defined and\\$\vdashr C : \sigma_C \to \rho_C$.
\item A finite set of mixins $\mathcal{M}$.
\item For each $M \in \mathcal{M}$ types $\sigma_M \in \TT$ and $\rho_M^1, \rho_M^2 \in \TTR$ such that $\interpret{\sigma_M}, \interpret{\rho_M^1}, \interpret{\rho_M^2}$ are defined and for all types $\rho \in \TTR$ we have 
$\vdashr M : (\sigma_M \to \rho \cap \rho_M^1) \to (\sigma_M \to \rho + \rho_M^2)$.
\item For each $M \in \mathcal{M}$ the non-empty set of labels $L_M = \lbl(\rho_M^2) \subseteq \mathcal{L}$ defined by $M$.
\end{itemize}

We translate given classes and mixins to the following repository $\DeltaLCM$ of combinators
\begin{align*}
\DeltaLCM = & \{C : \interpret{\sigma_C \to \rho_C} \mid C \in \mathcal{C}\} \\
& \cup \{M : \big((\interpret{\sigma_M} \to \interpret{\rho_M^1}) \to (\interpret{\sigma_M} \to \interpret{\rho_M^2})\big) \\
& \quad \cap \bigcap\limits_{l \in \calL \setminus L_M} \big((\interpret{\sigma_M} \to \rrecord{l(\alpha_l)}) \to (\interpret{\sigma_M} \to \rrecord{l(\alpha_l)})\big) \mid M \in \mathcal{M}\}
\end{align*}

Note that we use identifiers $C$ for classes and $M$ for mixins just as symbolic names in the repository, while they are also typable closed terms in $\lambdaR$.

To simplify notation, we introduce the infix metaoperator $\pipe$ such that $x \pipe f = f\ x$. It is left associative and has the lowest precedence. Accordingly, $x \pipe f \pipe g = g\ (f\ x)$.

Although types in $\DeltaLCM$ do not contain record-merge,  types of mixin compositions in $\bclkc$ are sound, which is shown in the following Theorem \ref{thm:soundness}. 

\begin{thm}[Soundness]
\label{thm:soundness}
Let $M_1, \ldots, M_n \in \mathcal{M}$ be mixins, let $C \in \mathcal{C}$ be a class, let $\sigma \in \TT, \rho \in \TTR$ be types such that $\interpret{\sigma \to \rho}$ is defined, and let $k \in \NN$.\\
If $\DeltaLCM \vdashk C \pipe M_1 \pipe \ldots \pipe M_n : \interpret{\sigma \to \rho}$, then $\vdashr C \pipe M_1 \pipe M_2 \pipe \ldots \pipe M_n : \sigma \to \rho$.
\end{thm}

\begin{proof}
Induction on $n$ using Lemma \ref{lem:enc_soundness} and $(\ArrE)$.
\end{proof}

Complementarily, types of mixin compositions in $\bclkc$ are complete wrt. $\lambdaR$ type assumptions for classes and mixins listed above. Specifically, we show in the following Theorem \ref{thm:partial_completeness} that, assuming arguably natural typings of classes and mixins, we can find corresponding $\bclkc$-counterparts of $\lambdaR$ type derivations.

\begin{thm}[Partial Completeness]
\label{thm:partial_completeness}
Let $\Gamma \subseteq \{x_C : \sigma_C \to \rho_C \mid C \in \mathcal{C}\} \cup \{x_M^\rho : (\sigma_M \to \rho \cap \rho_M^1) \to (\sigma_M \to \rho + \rho_M^2) \mid M \in \mathcal{M}, \rho \in \TTR, \interpret{\rho} \text{ is defined}\}$ be a finite context and let $\sigma \in \TT, \rho \in \TTR$ be types such that $\interpret{\sigma \to \rho}$ is defined. Let $k=  \max(\{\level(\interpret{\rho_C}) \mid C \in \calC\} \cup \{\level(\interpret{\rho_M^2}) \mid M \in \calM\} \cup \{\level(\interpret{\rho})\})$. Let $M_1, \ldots, M_n \in \mathcal{M}$ be mixins.\\
If $\Gamma \vdashr x_C \pipe x_{M_1}^{\rho_1} \pipe \ldots \pipe x_{M_n}^{\rho_n} : \sigma \to \rho$, then $\DeltaLCM \vdashk C \pipe M_1 \pipe \ldots \pipe M_n : \interpret{\sigma \to \rho}$.
\end{thm}

\begin{proof} Induction on $n$ choosing for each $x_{M}^{\rho}$ where $\rho = \bigcap\limits_{l \in L} \record{l : \tau_{l}}$ the substitution $\alpha_l \mapsto \interpret{\tau_{l}}$ for $l \in L \setminus L_M$ and $\alpha_l \mapsto \omega$ otherwise to type $M \in \DeltaLCM$ and using $(\ArrE)$, $(\leq)$, Lemma \ref{lem:enc_completeness}.
\end{proof}

\noindent Note that Theorem \ref{thm:partial_completeness} defines a bound $k$ based on the input. In the following, we use this bound ensuring that the provided examples are in fact computable. We extend the running example by the following mixin $\SuccTwice$.
\begin{align*}
&\SuccTwice = \lambda \argClass. \Y (\lambda \myClass. \lambda \state. \\
&\qquad \letin{\super = \argClass \ \state} \\
&\qquad \letin{\self = \myClass \ \state} \\
&\qquad \qquad \super \oplus \record{\succTwice = \letin{\super' = \argClass(\super.\succ)} \super'.succ})
\end{align*}
$\SuccTwice$ adds the method $\succTwice$ that is the twofold application of $\succ$. Note that the object $\super$ is updated functionally in $\succTwice$. Using the above translation, we obtain
\begin{align*}
&\Delta_{\{\get, \set, \succ, \succTwice, \compare\}}^{\{\Nat\},\{\SuccTwice, \Comparable\}} = \{\\
&\quad \Nat  :  \Int \to \rrecord{\get(\Int) \cap \set(\Int \to \Int) \cap \succ(\Int)}, \\
&\quad \Comparable  :  \big((\Int \to \rrecord{\get(\Int)}) \to (\Int \to \rrecord{\compare(\rrecord{\get(\Int)} \to \Bool)})\big)\\
&\qquad \cap \big((\Int \to \rrecord{\get(\alpha_\get)}) \to (\Int \to \rrecord{\get(\alpha_\get)})\big)\\
&\qquad \cap \big((\Int \to \rrecord{\set(\alpha_\set)}) \to (\Int \to \rrecord{\set(\alpha_\set)})\big)\\
&\qquad \cap \big((\Int \to \rrecord{\succ(\alpha_\succ)}) \to (\Int \to \rrecord{\succ(\alpha_\succ)})\big)\\
&\qquad \cap \big((\Int \to \rrecord{\succTwice(\alpha_\succTwice)}) \to (\Int \to \rrecord{\succTwice(\alpha_\succTwice)})\big),\\
&\quad \SuccTwice  :  \big((\Int \to \rrecord{\succ(\Int)}) \to (\Int \to \rrecord{\succTwice(\Int)})\big)\\
&\qquad \cap \big((\Int \to \rrecord{\get(\alpha_\get)}) \to (\Int \to \rrecord{\get(\alpha_\get)})\big)\\
&\qquad \cap \big((\Int \to \rrecord{\set(\alpha_\set)}) \to (\Int \to \rrecord{\set(\alpha_\set)})\big)\\
&\qquad \cap \big((\Int \to \rrecord{\succ(\alpha_\succ)}) \to (\Int \to \rrecord{\succ(\alpha_\succ)})\big)\\
&\qquad \cap \big((\Int \to \rrecord{\compare(\alpha_\compare)}) \to (\Int \to \rrecord{\compare(\alpha_\compare)})\big)\}
\end{align*}
We may ask inhabitation questions such as $$\Delta_{\{\get, \set, \succ, \succTwice, \compare\}}^{\{\Nat\},\{\SuccTwice, \Comparable\}} \vdashk ? : \interpret{\Int \to \record{\succ : \Int, \compare : \record{\get : \Int} \to \Bool, \succTwice : \Int}}$$
and obtain the combinatory term  \enquote{$\Nat \pipe \Comparable \pipe \SuccTwice$} as the synthesized result. From Theorem \ref{thm:soundness} we know $$ \vdashr \Nat \pipe \Comparable \pipe \SuccTwice : \Int \to \record{\succ : \Int, \compare : \record{\get : \Int} \to \Bool, \succTwice : \Int}$$


The presented encoding has several benefits with respect to scalability. 
First, the size of the presented repositories is polynomial in $|\mathcal{L}| * |\mathcal{C}| * |\mathcal{M}|$. 
Second, expanding the label set $\mathcal{L}$ can be performed automatically in polynomial time by adding additional components $(\interpret{\sigma_M} \to \rrecord{l(\alpha_l)}) \to (\interpret{\sigma_M} \to \rrecord{l(\alpha_l)})$ to each mixin $M$ for each new label $l$.
Third, adding a class/mixin to an existing repository is as simple as adding one typed combinator for the class/mixin. Existing combinators in the repository remain untouched. As an example, we add the following mixin $\SuccDelta$ to $\Delta_{\{\get, \set, \succ, \succTwice, \compare\}}^{\{\Nat\},\{\SuccTwice, \Comparable\}}$.
\begin{align*}
&\SuccDelta = \lambda \argClass. \Y (\lambda \myClass. \lambda \state. \\
&\qquad \letin{\super = \argClass \ \state} \\
&\qquad \letin{\self = \myClass \ \state} \super \oplus \record{\succ = \lambda d.(\super.\set(\super.get + d)}))
\end{align*}

In $\lambdaR$ for all types $\rho \in \TTR$ we have 
$$\vdashr \SuccDelta \colon (\Int \to \rho \cap \record{\get : \Int, \set : \Int \to \Int}) \to (\Int \to \rho + \record{\succ : \Int \to \Int})$$

We obtain the following extended repository
\begin{align*}
&\Delta_{\{\get, \set, \succ, \succTwice, \compare\}}^{\{\Nat\},\{\SuccTwice, \Comparable, \SuccDelta\}} = \Delta_{\{\get, \set, \succ, \succTwice, \compare\}}^{\{\Nat\},\{\SuccTwice, \Comparable\}} \\
& \qquad \cup \{ \SuccDelta : \big((\Int \to \rrecord{\get(\Int) \cap \set(\Int \to \Int)}) \to (\Int \to \rrecord{\succ(\Int \to \Int)})\big) \\
&\qquad \qquad \cap \big((\Int \to \rrecord{\get(\alpha_\get)}) \to (\Int \to \rrecord{\get(\alpha_\get)})\big)\\
&\qquad \qquad \cap \big((\Int \to \rrecord{\set(\alpha_\set)}) \to (\Int \to \rrecord{\set(\alpha_\set)})\big)\\
&\qquad \qquad \cap \big((\Int \to \rrecord{\succTwice(\alpha_\succTwice)}) \to (\Int \to \rrecord{\succTwice(\alpha_\succTwice)})\big)\\
&\qquad \qquad \cap \big((\Int \to \rrecord{\compare(\alpha_\compare)}) \to (\Int \to \rrecord{\compare(\alpha_\compare)})\big)\}
\end{align*}

Asking the inhabitation question $$\Delta_{\{\get, \set, \succ, \succTwice, \compare\}}^{\{\Nat\},\{\SuccTwice, \Comparable, \SuccDelta\}} \vdashk ? : \interpret{\Int \to \record{\succ : \Int \to \Int, \succTwice : \Int}}$$ synthesizes \enquote{$\Nat \pipe \SuccTwice \pipe \SuccDelta$}. Note that even in such a simplistic scenario the order in which mixins are applied can be crucial mainly because $\Override$ is not commutative. Moreover, the early binding of self and the associated preservation of overwritten methods allows for destructive overwrite of $\succ$ by $\SuccDelta$ without invalidating $\succTwice$ that is previously added by $\SuccTwice$. This also may make multiple applications of a single mixin meaningful.

In order to improve the reading experience, the running example is coherently arranged in the appendix.

\subsection{Semantic Types}\label{subsec:semanticTypes}
An additional benefit of using intersection types and, in particular, the (CL)S framework is the availability of the so called \emph{semantic types}~\cite{JR13}. Semantic types can be used to further specify the semantics of typed combinators in a repository and restrict/guide the inhabitant search. Semantic types consist of additional constants $\semType{a}$ and constructors $\semType{c}(\cdot)$ that arise in the current domain of interest (semantic types are well suited to capture taxonomies~\cite{JR13}). Native types are augmented by semantic types using intersection. For example, consider the native type $\Int$. We may be interested in whether the current value is $\Even$ or $\Odd$. Therefore, we may augment the native type by this information resulting in types of the form $\Int \cap \Even$ representing even integers or $\Int \cap \Odd$ representing odd integers. This becomes increasingly interesting when we have knowledge about functional dependencies between semantic types in the domain of interest. In particular, knowing that successors of even integers are odd and vice versa, we may augment the native type of $\Nat$ to represent this domain knowledge 
\begin{align*}
\Nat  : &  (\Int \cap \Even \to \rrecord{\get(\Int \cap \Even)}) \\
& \cap (\Int \cap \Odd \to \rrecord{\get(\Int \cap \Odd)}) \\
&\cap (\Int \to \rrecord{\set((\Int \to \Int) \cap (\Even \to \Even) \cap (\Odd \to \Odd))}) \\
&\cap (\Int \cap \Even \to \rrecord{\succ(\Int \cap \Odd)}) \\
&\cap (\Int \cap \Odd \to \rrecord{\succ(\Int \cap \Even)})
\end{align*}

The above type of $\Nat$ expresses that the $\get$ method returns the semantic type of the underlying state of the object, the $\set$ method does not take the underlying state of the object into account, and the $\succ$ method returns the opposed semantic type. Note that, if all semantic components are erased, the remaining type of $\Nat$ is exactly the original native type.

In the setting of mixin composition synthesis, semantic types enhance the expressiveness of the type system. The mixin $\SuccTwice$ from our running example can be typed in the following way

\begin{align*}
\SuccTwice : (\Int \cap \Even \to \rrecord{\succ(\Int \cap \Odd)}) \cap (\Int \cap \Odd \to \rrecord{\succ(\Int \cap \Even)}) \\
\quad \to (\Int \cap \Even \to \rrecord{\succTwice(\Int \cap \Even)}) \cap (\Int \cap \Odd \to \rrecord{\succTwice(\Int \cap \Odd)})
\end{align*}

The above type of $\SuccTwice$ expresses that given the proper semantic types of $\succ$, the semantic type of $\succTwice$ corresponds to the semantic type of the underlying state of the object, i.e.~the twofold successor of an even (resp. odd) integer is even (resp. odd).

This allows to distinguish methods with different semantics in the domain of interest but exhibiting identical native types. Consider the following mixin $\Parity$ that overwrites the method $\succ$ to be the twofold successor such that the parity of the underlying state remains the same. 
\begin{align*}
&\Parity = \lambda \argClass. \Y (\lambda \myClass. \lambda \state. \\
&\qquad \letin{\super = \argClass \ \state} \\
&\qquad \letin{\self = \myClass \ \state} \super \oplus \record{\succ = \super.\succTwice})
\end{align*}
$\Parity$ can be typed in the following way
\begin{align*}
\Parity : (\Int \cap \Even \to \rrecord{\succTwice(\Int \cap \Even)}) \cap (\Int \cap \Odd \to \rrecord{\succTwice(\Int \cap \Odd)}) \\
\quad \to (\Int \cap \Even \to \rrecord{\succ(\Int \cap \Even)}) \cap (\Int \cap \Odd \to \rrecord{\succ(\Int \cap \Odd)})
\end{align*}

Asking the inhabitation question 
$$\Delta_{\{\get, \set, \succ, \succTwice, \compare\}}^{\{\Nat\},\{\SuccTwice, \Comparable, \SuccDelta, \Parity\}} \vdashk \Int \cap \Even \to \rrecord{\succ : \Int \cap \Even}$$ results in \enquote{$\Nat \pipe \SuccTwice \pipe \Parity$}. Note that due to the early binding, $\succTwice$ uses the old $\succ$ method which has the required semantics in order for $\SuccTwice$ to be applied.


Let us explore the descriptive capabilities of semantic types using a more illustrative example. Consider the domain of cryptography containing algorithms for encrypting and signing data. Describing abstract properties such as \enquote{encrypted} or \enquote{signed} at the lowest level, e.g.~using Hoare logic, requires an enormous amount of work. In practice, such properties are described textually while the native type, e.g.~$\String$, does not capture particular semantics. Consider the following repository
\begin{align*}
\Delta_{\text{native}} = \{ 
& \text{Reader} : && \String \to \rrecord{\get(\String)}, \\
& \text{Enc} : && (\String \to \rrecord{\get(\String)}) \to \String \to \rrecord{\get(\String)}, \\
& \text{Sign} : && (\String \to \rrecord{\get(\String)}) \to \String \to \rrecord{\get(\String)}, \\
& \text{Time} : && (\String \to \rrecord{\get(\String)}) \to \String \to \rrecord{\get(\String)} \}
\end{align*}
with the following textual description
\begin{itemize}
\item The Reader class provides a $\get$ method that returns some plain text data.
\item The Enc mixin replaces the $\get$ method of a given class with a new $\get$ method that returns the encrypted result of the overwritten $\get$ method.
\item The Sign mixin replaces the $\get$ method appending the signature to the result of the overwritten $\get$ method.
\item The Time mixin replaces the $\get$ method appending a time-stamp to the result of the overwritten $\get$ method.
\end{itemize}

Unfortunately, the types in $\Delta$ are too general to be used for synthesis of meaningful compositions. However, we can use semantic types to embed the textual description of the particular semantics in our domain of interest into $\Delta_{\text{native}}$ resulting in the following repository
\begin{align*}
\Delta = \{ 
& \text{Reader} : && \String \to \rrecord{\get(\String \cap \Plain)}, \\
& \text{Enc} : && (\String \to \rrecord{\get(\String \cap \alpha)}) \to (\String \to \rrecord{\get(\String \cap \Enc(\alpha))}), \\
& \text{Sign} : && (\String \to \rrecord{\get(\String \cap \alpha)}) \to (\String \to \rrecord{\get(\String \cap \alpha \cap \Sign(\alpha))}), \\
& \text{Time} : && (\String \to \rrecord{\get(\String \cap \alpha)}) \to (\String \to \rrecord{\get(\String \cap \alpha \cap \Time)}) \}
\end{align*}
In the above repository $\Delta$, the mixin Enc replaces any semantic information $\alpha$ of the $\get$ method by $\Enc(\alpha)$. The mixin Sign adds semantic information $\Sign(\alpha)$ to any previous semantic information $\alpha$ while also preserving $\alpha$. The mixin Time adds semantic information $\Time$ while preserving the old semantic information.

If we are interested in a composition that has a $\get$ method returning an encrypted plain text with a time-stamp and a signature, we may ask the following inhabitation question
$$\Delta \vdashk ? : \String \to \rrecord{get(\String \cap \Enc(\Plain \cap \Time \cap \Sign(\Plain \cap \Time)))}$$
The above question is answered by \enquote{$\text{Reader} \pipe \text{Time} \pipe \text{Sign} \pipe \text{Enc}$}.
If we are interested in a composition that encrypts the plain-text thrice, we may ask the following inhabitation question
$$\Delta \vdashk ? : \String \to \rrecord{get(\String \cap \Enc(\Enc(\Enc(\Plain))))}$$
The above question is answered by \enquote{$\text{Reader} \pipe \text{Enc} \pipe \text{Enc} \pipe \text{Enc}$}.


	\section{Conclusion and Future Work}\label{sec:conclusion}
	
We presented a theory for automatic compositional construction of object-oriented classes by combinatory synthesis. This theory is based on the $\lambda$-calculus with records and record merge $\Override$ typed by intersection types with records and $+$. It is capable of modeling classes as functions from states to records (i.e.~objects), and mixins as functions from classes to classes. Mixins can be assigned meaningful types using $+$ expressing their compositional character. However, non-monotonic properties of $+$ are incompatible with the existing theory of \bcl\, synthesis. Therefore, we designed a translation to repositories of combinators typed in $\bclc$. We have proven this translation to be sound (Theorem \ref{thm:soundness}) and partially complete (Theorem \ref{thm:partial_completeness}). A notable feature is the encoding of negative information (the absence of labels). The original approach~\cite{BDDCdLR15} exploited the logic programming capabilities of inhabitation by adding sets of combinators serving as witnesses for the non-presence of labels. In this work we further refined the encoding by embedding the distinction between labels that are accessed, modified or untouched by a mixin into the combinator type directly. In Section~\ref{sec:synthesis} we also showed that this encoding scales wrt.~extension of repositories.

Future work includes further studies on the possibilities to encode predicates exploiting patterns similar to the negative information encoding. 
Another direction of future work is to extend types of mixins and classes by semantic as well as modal types \cite{DudderMR14}, a development initiated in \cite{BessaiDDM14}. In particular, the expressiveness of semantic types can be used to assign meaning to multiple applications of a single mixin and allow to reason about object-oriented code on a higher abstraction level as well as with higher semantic accuracy.

%


\bibliographystyle{alpha}
\bibliography{SemMixBiblio,SyntBibliography}

\newpage
\appendix
\section{}
Proof of Lemma \ref{lem:ttc_path_intersection}.
\def\atPosition#1#2{ASD}
\def\leqC{ASD}
\begin{proof}
By definition, $\bigcap \PP(\tau)$ is $\omega$ or an intersection of paths, therefore it is organized. Due to the identities $\sigma \to \omega = \omega$ and $\omega \cap \omega = \omega$ we have that if $\PP(\tau) = \emptyset$ then $\tau = \omega$.
We show by induction on the depth of the syntax tree of $\tau$ that $\bigcap \PP(\tau) = \tau$.
\begin{description}
\item[Basis Step]
If $\tau \equiv \omega$ or $\tau \equiv \alpha$ or $\tau \equiv a$, then $\bigcap \PP(\tau) = \tau$ by definition.
\item[Inductive Step]\ 
\begin{description}
\item[Case 1.1] $\tau \equiv \tau_1 \to \tau_2$ and $\PP(\tau_2) = \emptyset$. \\
Since $\PP(\tau_2) = \emptyset$, we have $\tau_2 = \omega$, therefore $\tau = \omega = \bigcap \emptyset = \bigcap \PP(\tau)$.
\item[Case 1.2] $\tau \equiv \tau_1 \to \tau_2$ and $\PP(\tau_2) \neq \emptyset$. \\
$\bigcap \PP(\tau) = \bigcap \{\tau_1 \to \pi \mid \pi \in \PP(\tau_2)\} = \tau_1 \to \bigcap \PP(\tau_2) \stackrel{\text{IH}}{=} \tau_1 \to \tau_2 = \tau$.
\item[Case 2] $\tau \equiv \tau_1 \cap \tau_2$. \\
$\bigcap \PP(\tau) = \bigcap \PP(\tau_1) \cap \bigcap \PP(\tau_2) \stackrel{\text{IH}}{=} \tau_1 \cap \tau_2 = \tau$.
\item[Case 3.1] $\tau \equiv c(\tau_1)$ and $\PP(\tau_1) = \emptyset$.\\
$\bigcap \PP(\tau) = c(\omega) = \tau$.
\item[Case 3.2] $\tau \equiv c(\tau_1)$ and $\PP(\tau_1) \neq \emptyset$.\\
$\bigcap \PP(\tau) = c(\bigcap \PP(\tau_1)) \stackrel{\text{IH}}{=} c(\tau_1) = \tau$
\end{description}
\end{description}
\end{proof}

Proof of Lemma \ref{lem:ttc_path_set_subtyping}.
\begin{proof}
For \enquote{$\Leftarrow$} we have $\sigma \stackrel{\text{Lem. } \ref{lem:ttc_path_intersection}}{=} \bigcap \PP(\sigma) \leq \bigcap \PP(\tau) \stackrel{\text{Lem. } \ref{lem:ttc_path_intersection}}{=} \tau$ using the fact that $\sigma_1 \leq \tau_1$  and $\sigma_2 \leq \tau_2$ implies $\sigma_1 \cap \sigma_2 \leq \tau_1 \cap \tau_2$.
For \enquote{$\Rightarrow$} we show by induction on the derivation of $\sigma \leq \tau$ that for each path $\pi \in \PP(\tau)$ there exists a path $\pi' \in \PP(\sigma)$ such that $\pi' \leq \pi$ and the desired properties hold.
\begin{description}
\item[Basis Step] In the following cases we have $\PP(\sigma) \supseteq \PP(\tau)$, therefore the desired properties hold via syntactical identity
\begin{itemize}
\item $\sigma \equiv \sigma \leq \sigma \equiv \tau$
\item $\sigma \equiv \sigma \leq \omega \equiv \tau$
\item $\sigma \equiv \omega \leq \omega \to \omega \equiv \tau$
\item $\sigma \equiv \sigma' \cap \tau' \leq \sigma' \equiv \tau$ for some $\sigma', \tau' \in \TTC$
\item $\sigma \equiv \sigma' \cap \tau' \leq \tau' \equiv \tau$ for some $\sigma', \tau' \in \TTC$
\item $\sigma \equiv (\sigma' \to \tau_1') \cap (\sigma' \to \tau_2') \leq \sigma' \to \tau_1' \cap \tau_2' \equiv \tau$ for some $\sigma', \tau_1', \tau_2' \in \TTC$
\item $\sigma \equiv c(\sigma_1) \cap c(\tau_1) \leq c(\sigma_1 \cap \tau_1) \equiv \tau$ for some $\sigma_i, \tau_i \in \TTC$ for $1 \leq i \leq n$
\end{itemize}
\item[Inductive Step] Let $\pi \in \PP(\tau)$.
\begin{description}
\item[Case 1] $\sigma \equiv \sigma_1 \to \tau_1 \leq \sigma_2 \to \tau_2 \equiv \tau$ with $\sigma_2 \leq \sigma_1$ and $\tau_1 \leq \tau_2$.\\
We have $\PP(\sigma_2 \to \tau_2) \ni \pi \equiv \sigma_2 \to \pi_2$ for some path $\pi_2 \in \PP(\tau_2)$. By the induction hypothesis there exists a path $\pi_1 \in \PP(\tau_1)$ such that $\pi_1 \leq \pi_2$. Therefore, $$\PP(\sigma) \ni \sigma_1 \to \pi_1 \leq \sigma_2 \to \pi_2 \equiv \pi \text{ with } \sigma_2 \leq \sigma_1 \text{ and } \pi_1 \leq \pi_2$$

\item[Case 2] $\sigma \equiv \sigma \leq \tau_1 \cap \tau_2 \equiv \tau$ with $\sigma \leq \tau_1$ and $\sigma \leq \tau_2$.\\
Since $\pi \in \PP(\tau) = \PP(\tau_1) \cup \PP(\tau_2)$, we have $\pi \in \PP(\tau_1)$ or $\pi \in \PP(\tau_2)$. By the induction hypothesis in both cases there exists a path $\pi' \in \PP(\sigma)$ such that $\pi' \leq \pi$ satisfying the desired properties.

\item[Case 3] $\sigma \equiv c(\sigma_1) \leq c(\tau_1) \equiv \tau$ with $\PP(\tau_i)=\emptyset$.\\
Since $\pi \equiv c(\omega)$, we have $\pi' \leq \pi$ satisfying the desired properties for any $\pi' \in \PP(\sigma)$, noting that $\PP(\sigma) \neq \emptyset$.

\item[Case 4] $\sigma \equiv c(\sigma_1) \leq c(\tau_1) \equiv \tau$ with $\sigma_1 \leq \tau_1$ and $\PP(\tau_1) \neq \emptyset$.\\
By Definition \ref{def:ttc_paths_in_tau}, we have $\pi \equiv c(\pi_\tau)$ for some $\pi_\tau \in \PP(\tau_1)$.
By the induction hypothesis there exists a $\pi_\sigma \in \PP(\sigma_1)$ such that $\pi_\sigma \leq \pi_\tau$. Therefore, $\PP(\sigma) \ni c(\pi_\sigma) \leq c(\pi_\tau) \equiv \pi$.

\item[Case 5] $\sigma \leq \sigma_1 \leq \tau$.\\
By the induction hypothesis there exists a path $\pi_1 \in \PP(\sigma_1)$ such that $\pi_1 \leq \pi$. Again by the induction hypothesis there exists a path $\pi' \in \PP(\sigma)$ such that $\pi' \leq \pi_1$. Therefore, we have $\PP(\sigma) \ni \pi' \leq \pi_1 \leq \pi$.
\end{description}
\end{description}
\end{proof}

Proof of Lemma \ref{lem:blckc_path_lemma}.
\begin{proof}
$(2) \Rightarrow (1)$ is trivial using the rules $(\ArrE), (\leq)$ and the property $(\sigma_1 \to \tau_1) \cap (\sigma_2 \to \tau_2) \leq \sigma_1 \cap \sigma_2 \to \tau_1 \cap \tau_2$.
For $(1) \Rightarrow (2)$, we use induction on the derivation of $\Delta \vdashk C \, E_1 \ldots E_m : \tau$. Wlog. $\tau \neq \omega$ and (similar to~\cite[Proposition 8]{RehofEtAlCSL12}) the derivation $\Delta \vdashk C \, E_1 \ldots E_m : \tau$ does not contain atoms other than $\atoms(\Delta) \cup \atoms(\tau)$.
\begin{description}
\item[Basis Step] The last rule is $\prooftree
	C:\sigma \in \Delta \qquad \level(S) \leq k
	\justifies
	\Delta \vdashk C:S(\sigma) \equiv \tau
	\using (\text{Var})
\endprooftree$. For $P = \PP_0(S(\sigma))$ we have $\bigcap_{\pi \in P} \tgt_0(\pi) = \bigcap P = \bigcap \PP_0(S(\sigma)) = \bigcap \PP(S(\sigma)) \stackrel{\text{Lem. } \ref{lem:ttc_path_intersection}}{=} S(\sigma) = \tau$.
\item[Inductive Step]\ \\
\begin{description}
\item[Case 1] The last rule is $\prooftree
	\Delta \vdashk C \, E_1 \ldots E_m: \sigma \qquad \sigma\leq\tau
	\justifies
	\Delta \vdashk C \, E_1 \ldots E_m : \tau
	\using (\leq)
\endprooftree $.\\ 
Follows immediately from the induction hypothesis with $\bigcap_{\pi \in P} \tgt_m(\pi) \leq \sigma \leq \tau$.
\item[Case 2] The last rule is $\prooftree
	\Delta \vdashk C \, E_1 \ldots E_m: \tau_1 \quad 
	\Delta \vdashk C \, E_1 \ldots E_m: \tau_2
	\justifies
	\Delta \vdashk C \, E_1 \ldots E_m: \tau_1\inter\tau_2 \equiv \tau
	\using (\inter)
\endprooftree$.\\
By the induction hypothesis, there exist sets\\
$P_1, P_2 \subseteq \PP_m(\bigcap \{S(\Delta(C)) \mid \level(S) \leq k, \atoms(S) \subseteq \atoms(\Delta) \cup \atoms(\tau)\})$\\ such that $\bigcap_{\pi \in P_j} \tgt_m(\pi) \leq \tau_j$ and $\Delta \vdashk E_i : \bigcap_{\pi \in P_j} \arg_i(\pi)$ for $1 \leq i \leq m$ and $j \in \{1,2\}$.
For $P = P_1 \cup P_2$ we obtain 
\begin{enumerate}
\item $\bigcap_{\pi \in P} \tgt_m(\pi) = \bigcap_{\pi \in P_1} \tgt_m(\pi) \cap \bigcap_{\pi \in P_2} \tgt_m(\pi) \leq \tau_1 \cap \tau_2 = \tau$
\item $\Delta \vdashk E_i : \bigcap_{\pi \in P} \arg_i(\pi) = \bigcap_{\pi \in P_1} \arg_i(\pi) \cap \bigcap_{\pi \in P_2} \arg_i(\pi)$ for $1 \leq i \leq m$ using the induction hypothesis and the rule $(\inter)$.
\end{enumerate}
\item[Case 3] The last rule is $\prooftree
	\Delta \vdashk C \, E_1 \ldots E_{m-1}: \sigma\to\tau \qquad \Delta \vdashk E_m:\sigma
	\justifies
	\Delta \vdashk C \, E_1 \ldots E_m:\tau
	\using (\ArrE)
\endprooftree$.\\
By the induction hypothesis, there exists a set\\
 $P' \subseteq \PP_{m-1}(\bigcap \{S(\Delta(C)) \mid \level(S) \leq k, \atoms(S) \subseteq \atoms(\Delta) \cup \atoms(\tau)\})$\\
  such that $\bigcap_{\pi \in P'} \tgt_{m-1}(\pi) \leq \sigma \to \tau$. Let $P = \{\pi \in P' \mid \sigma \leq \arg_m(\pi)\}$. By Lemma \ref{lem:ttc_path_set_subtyping} for each path $\sigma \to \pi \in \PP(\sigma \to \tau)$ there exists a path $\pi' \in P'$ such that $\tgt_{m-1}(\pi') \leq \sigma \to \pi$ and $\tgt_{m-1}(\pi') = \arg_{m}(\pi') \to \tgt_m(\pi')$ with $\sigma \leq \arg_m(\pi')$ and $\tgt_m(\pi') \leq \pi$, therefore $\pi' \in P$. By Lemma \ref{lem:ttc_path_set_subtyping} we obtain $\bigcap_{\pi \in P} \tgt_{m}(\pi) \leq \tau$. Additionally, by the induction hypothesis $\Delta \vdashk E_i : \bigcap_{\pi \in P'} \arg_i(\pi) \leq \bigcap_{\pi \in P} \arg_i(\pi)$ for $1 \leq i \leq m-1$, and $\Delta \vdashk E_m : \sigma \leq \bigcap_{\pi \in P} \arg_m(\pi)$.
\end{description}
\end{description}
\end{proof}


\section{Running Example}\label{app:running}
\noindent Overview of terms
\begin{align*}
&\Nat = \Y (\lambda \myClass. \lambda \state. \\
&\qquad \letin{\self = \myClass \ \state} \\
&\qquad \qquad \record{\get = \state, \set = \lambda \State'.\state', \succ = \self.\set(\self.\get + 1)}\\
&\Comparable = \lambda \argClass. \Y (\lambda \myClass. \lambda \state. \\
&\qquad \letin{\super = \argClass \ \state} \\
&\qquad \letin{\self = \myClass \ \state} \\
&\qquad \qquad \super \oplus \record{\compare = \lambda o.(o.\get == \self.\get)})\\
&\SuccTwice = \lambda \argClass. \Y (\lambda \myClass. \lambda \state. \\
&\qquad \letin{\super = \argClass \ \state} \\
&\qquad \letin{\self = \myClass \ \state} \\
&\qquad \qquad \super \oplus \record{\succTwice = \letin{\super' = \argClass(\super.\succ)} \super'.succ})\\
&\SuccDelta = \lambda \argClass. \Y (\lambda \myClass. \lambda \state. \\
&\qquad \letin{\super = \argClass \ \state} \\
&\qquad \letin{\self = \myClass \ \state} \\
&\qquad \qquad \super \oplus \record{\succ = \lambda d.(\super.\set(\super.get + d)}))\\
&\Parity = \lambda \argClass. \Y (\lambda \myClass. \lambda \state. \\
&\qquad \letin{\super = \argClass \ \state} \\
&\qquad \letin{\self = \myClass \ \state} \\
&\qquad \qquad \super \oplus \record{\succ = \super.\succTwice})
\end{align*}

Overview of $\lambdaR$ types for any record type $\rho \in \TTR$
\begin{align*}
&\Nat  :  \Int \to \record{\get : \Int, \set : \Int \to \Int, \succ : \Int} \\
&\Comparable  :  (\Int \to \rho \cap \record{\get : \Int}) \to (\Int \to \rho + \record{\compare : \record{\get : \Int} \to \Bool})\\
&\SuccTwice  :  (\Int \to \rho \cap \record{\succ : \Int}) \to (\Int \to \rho + \record{\succTwice : \Int})\\
&\SuccDelta : (\Int \to \rho \cap \record{\get : \Int, \set : \Int \to \Int}) \to (\Int \to \rho + \record{\succ : \Int \to \Int})\\
&\Parity : (\Int \to \rho \cap \record{\succTwice : \Int}) \to (\Int \to \rho + \record{\succ : \Int})
\end{align*}

\bigskip

\noindent Overview of $\TTC$ types
\begin{align*}
&\Nat  :  \Int \to \rrecord{\get(\Int) \cap \set(\Int \to \Int) \cap \succ(\Int)}, \\
&\Comparable  :  \big((\Int \to \rrecord{\get(\Int)}) \to (\Int \to \rrecord{\compare(\rrecord{\get(\Int)} \to \Bool)})\big)\\
&\qquad \cap \big((\Int \to \rrecord{\get(\alpha_\get)}) \to (\Int \to \rrecord{\get(\alpha_\get)})\big)\\
&\qquad \cap \big((\Int \to \rrecord{\set(\alpha_\set)}) \to (\Int \to \rrecord{\set(\alpha_\set)})\big)\\
&\qquad \cap \big((\Int \to \rrecord{\succ(\alpha_\succ)}) \to (\Int \to \rrecord{\succ(\alpha_\succ)})\big)\\
&\qquad \cap \big((\Int \to \rrecord{\succTwice(\alpha_\succTwice)}) \to (\Int \to \rrecord{\succTwice(\alpha_\succTwice)})\big)\\
&\SuccTwice  :  \big((\Int \to \rrecord{\succ(\Int)}) \to (\Int \to \rrecord{\succTwice(\Int)})\big)\\
&\qquad \cap \big((\Int \to \rrecord{\get(\alpha_\get)}) \to (\Int \to \rrecord{\get(\alpha_\get)})\big)\\
&\qquad \cap \big((\Int \to \rrecord{\set(\alpha_\set)}) \to (\Int \to \rrecord{\set(\alpha_\set)})\big)\\
&\qquad \cap \big((\Int \to \rrecord{\succ(\alpha_\succ)}) \to (\Int \to \rrecord{\succ(\alpha_\succ)})\big)\\
&\qquad \cap \big((\Int \to \rrecord{\compare(\alpha_\compare)}) \to (\Int \to \rrecord{\compare(\alpha_\compare)})\big)\\
&\SuccDelta : \big((\Int \to \rrecord{\get(\Int) \cap \set(\Int \to \Int)}) \to (\Int \to \rrecord{\succ(\Int \to \Int)})\big)\\
&\qquad \cap \big((\Int \to \rrecord{\get(\alpha_\get)}) \to (\Int \to \rrecord{\get(\alpha_\get)})\big)\\
&\qquad \cap \big((\Int \to \rrecord{\set(\alpha_\set)}) \to (\Int \to \rrecord{\set(\alpha_\set)})\big)\\
&\qquad \cap \big((\Int \to \rrecord{\succTwice(\alpha_\succTwice)}) \to (\Int \to \rrecord{\succTwice(\alpha_\succTwice)})\big)\\
&\qquad \cap \big((\Int \to \rrecord{\compare(\alpha_\compare)}) \to (\Int \to \rrecord{\compare(\alpha_\compare)})\big)\\
&\Parity : \big((\Int \to \rrecord{\succTwice(\Int)}) \to (\Int \to \rrecord{\succ(\Int)})\big)\\
&\qquad \cap \big((\Int \to \rrecord{\get(\alpha_\get)}) \to (\Int \to \rrecord{\get(\alpha_\get)})\big)\\
&\qquad \cap \big((\Int \to \rrecord{\set(\alpha_\set)}) \to (\Int \to \rrecord{\set(\alpha_\set)})\big)\\
&\qquad \cap \big((\Int \to \rrecord{\succTwice(\alpha_\succTwice)}) \to (\Int \to \rrecord{\succTwice(\alpha_\succTwice)})\big)\\
&\qquad \cap \big((\Int \to \rrecord{\compare(\alpha_\compare)}) \to (\Int \to \rrecord{\compare(\alpha_\compare)})\big)
\end{align*}

%

\end{document}